\documentclass[11pt,reqno,notitlepage]{article}
\usepackage{setspace}
\usepackage{amssymb}
\usepackage{amsthm}
\usepackage{graphicx}
\usepackage{amsmath}
\usepackage{fancyhdr}
\usepackage{datetime}
\usepackage[toc]{appendix}
\usepackage[usenames,dvipsnames]{color}
\usepackage{tikz}
\usepackage{mathrsfs}
\usepackage[normalem]{ulem}
\usepackage{graphicx, xcolor}
\usepackage{hyperref}

\usepackage{enumerate}

\usepackage{setspace}
\usepackage{graphicx, xcolor}												% need for figures
\usepackage[mathscr]{eucal}												% gives you \mathscr font
\usepackage{subcaption}
\usepackage{color}
\usepackage{epstopdf}

\allowdisplaybreaks

\makeatletter
\newcommand{\labitem}[2]{%
\def\@itemlabel{\textbf{#1}}
\item
\def\@currentlabel{#1}\label{#2}}
\makeatother

\usetikzlibrary{arrows,decorations.pathreplacing}
\usepackage[top=1in, bottom=1in, left=1in, right=1in]{geometry}
\usepackage{mathtools}									%
\mathtoolsset{showonlyrefs=true}							%

\newtheorem{theorem}{Theorem}
\newtheorem{lemma}[theorem]{Lemma}								%
\newtheorem{proposition}[theorem]{Proposition}	
\newtheorem{corollary}[theorem]{Corollary}	
	
\newtheorem{definition}[theorem]{Definition}
\newtheorem{remark}{Remark}

\numberwithin{equation}{section}	
\numberwithin{theorem}{section}
\renewcommand{\vec}[1]{\mathbf{#1}}

\def\E{\mathbb{E}}			
\def\P{\mathbb{P}}										

\def\Q{\mathbb{Q}}
\def\R{\mathbb{R}}

\def\A{\mathcal{A}}

\def\F{\mathcal{F}}

\def\B{\mathcal{B}}

\def\M{\mathcal{M}}

\def\d{\partial}		

\def\ind{\mathbb{I}}

\newcommand{\e}[1]{\operatorname{e}^{#1}}

\newcommand{\dom}{\operatorname{dom}}
\newcommand{\domu}{L^\infty_{*}}

\renewcommand{\vec}[1]{\mathbf{#1}}

\def\KC{\mathrm{KC}}
\def\PHI{\mathrm{PHI}}

\DeclareMathOperator*{\argmax}{arg\,max}

\DeclareMathOperator*{\esssup}{ess\,sup}
\DeclareMathOperator*{\essinf}{ess\,inf}

\newtheorem{example}[theorem]{Example}

\long\def\symbolfootnote[#1]#2{\begingroup\def\thefootnote{\fnsymbol{footnote}}\footnote[#1]{#2}\endgroup}

\begin{document}

\title{Decentralized Prediction Markets and Sports Books} 

\author{
Hamed Amini
\thanks{
University of Florida, Department of Industrial and Systems Engineering, Gainesville, FL 32611, USA, email: {\tt aminil@ufl.edu}.
}
\and Maxim Bichuch
\thanks{
Department of Mathematics,
SUNY at Buffalo,
Buffalo, NY 14260, USA. 
{\tt mbichuch@buffalo.edu}. 
%Research is partially supported by Stellar Development Foundation Academic Research Grants program. %Work  is partially supported by NSF grant DMS-1736414.
}
\and Zachary Feinstein
\thanks{
School of Business,
Stevens Institute of Technology,
Hoboken, NJ 07030, USA,
{\tt  zfeinste@stevens.edu}. 
%The author acknowledges the support from NSF IUCRC CRAFT center research grant (2113906) for this research. The opinions expressed in this publication do not necessarily represent the views of NSF IUCRC CRAFT.
}
}
\date{\today}
\maketitle

\begin{abstract}
Prediction markets allow traders to bet on potential future outcomes. These markets exist for weather, political, sports, and economic forecasting. Within this work we consider a decentralized framework for prediction markets using automated market makers (AMMs). Specifically, we construct a liquidity-based AMM structure for prediction markets that, under reasonable axioms on the underlying utility function, satisfy meaningful financial properties on the cost of betting and the resulting pricing oracle. Importantly, we study how liquidity can be pooled or withdrawn from the AMM and the resulting implications to the market behavior. In considering this decentralized framework, we additionally propose financially meaningful fees that can be collected for trading to compensate the liquidity providers for their vital market function.  
\end{abstract}

%{\bf AMS subject classification} \\
%\indent {\bf JEL subject classification} \\

\section{Introduction}\label{sec:intro}

\subsection{Motivation}\label{sec:intro-motivation}

Decentralized Finance (DeFi) -- the novel paradigm utilizing blockchain for financial intermediation -- has the opportunity to democratize finance insofar as it opens the positions of financial intermediaries to individual investors. Prior to the ``crypto winter'' in May 2022, the value of all DeFi projects reached a high of \$180B.\footnote{\url{https://tradingview.com/markets/cryptocurrencies/global-charts/}} Automated market makers (AMMs) offer a prominent example of a successful DeFi application; these entities create spot markets between digital assets by holding liquidity pools in an equilibrium. As a decentralized market, these AMMs allow investors to pool their own assets into the market in exchange for a fraction of the fees being collected from trading with the pool.

Though AMMs are now regarded as a key component of DeFi, the idea of an automated market maker was first proposed for use in prediction markets (see the literature review in Section~\ref{sec:intro-lit} below or \cite[Section 1.1]{schlegel2022axioms}). However, in the context of prediction markets, such AMMs have been written as centralized markets only. The goal of this paper is to revisit prediction market AMMs with an emphasis on how to decentralize behaviors. In doing so, we seek to unify the concepts of AMMs for spot markets and prediction markets. In particular, the decentralized paradigm requires careful consideration for how to accommodate changing liquidity at the AMM (to permit investors to pool or withdraw liquidity) between the opening of the market and prior to the revelation of the realized event.

More specifically, once we define the notion of the prediction market, we want to consider the properties of these AMMs, i.e., for the cost assessed to bettors and the quoted pricing measure. Though many of these properties have been studied previously in the centralized market setting (see the literature review in Section~\ref{sec:intro-lit} below), all works the authors' are aware of require a finite set of possible outcomes and without any explicit consideration for the impact of liquidity on market behavior. Herein, we generalize such markets to allow for bets to be placed on general probability spaces with the explicit dependence on the cash re

Furthermore, due to the possibility of investing in the AMM, it is important to quantify the fees charged to traders. As highlighted in, e.g.,~\cite{bichuch2022axioms}, naively defining transaction costs can lead to unforeseen consequences to the profits of the AMM. Therefore, to avoid obvious mis-pricing, we also investigate these issues within the setting of prediction markets; prior works avoid these issues due to the centralized nature of the market.

Lastly, by introducing the mathematical construction for decentralized prediction markets, we believe that new products can be introduced to allow for, e.g., decentralized sports books which operate at a fraction of the cost of centralized sports books. Therefore, throughout this work we keep an eye towards the practicality of implementation and costs for these operations.

\subsection{Literature review}\label{sec:intro-lit}

Prediction markets seek to elicit probabilities of events through the aggregation of investor beliefs \cite{chen2013cost}. If an investor has a different belief of the likelihood of an event he or she can transact to capture that discrepancy which moves the quoted probabilities. These markets can be applied to weather forecasting (e.g., \cite{murphy1984probability}) or economic forecasting (e.g., \cite{o1977subjective}). However, even in a finite state space, Bayesian updating of these probabilities can be analytically and computationally challenging (see, e.g., \cite{aumann2016agreeing,geanakoplos1982we,hanson2002disagreement}).
One approach taken in practice is to use a scoring rule; repeated games with the scoring rule will lead to a Nash equilibrium \cite{kalai1993rational} and thus convergence to a common estimator for the market participants \cite{monderer1989approximating,hanson2003bayesian}.

In practice, financial \cite{leroy1998market}, sports \cite{hausch2008efficiency} and election prediction markets perform well. However, these markets can sometimes result in irrational prices, especially when market liquidity is low. Choosing the correct market structure, i.e., a scoring rule with appropriate properties such as the logarithmic scoring rule of \cite{hanson2007logarithmic} or a Bayesian framework of \cite{dai2021wisdom}, can improve market performance. Hanson's logarithmic scoring rule has been implemented in a dynamic prediction market and achieved good results~\cite{othman2010automated}. 
It has been generalized, and characterized as utility functions, in~\cite{chen2012utility}. Furthermore, that paper highlights an additional property for a market to have, i.e., bounded loss. \cite{othman2013practical} emphasizes other desirable properties for prediction markets, e.g., liquidity sensitivity.
We also refer the reader to \cite{wolfers2006prediction,tziralis2007prediction} for a survey of literature on prediction markets.

These predictions market scoring rules form AMMs. More recently, AMMs have been used prominently as \emph{decentralized} markets for digital assets. These decentralized AMMs were described in the early whitepapers of \cite{uniswapv1,uniswapv2}. Those whitepapers emphasize the key mathematical and algorithmic components of AMMs. 
One such idea that provides an algorithm for an AMM to decide on quantities to swap is the constant function, see e.g.~\cite{angeris2020improved,lipton2021automated,angeris2021constant,clark2020replicating,cartea2022decentralised,capponi2021adoption}. The structure for many decentralized AMMs was provided in \cite{xu2021sok,angeris2020does}. These constructions were generalized and given axiomatic foundations in~\cite{schlegel2022axioms,bichuch2022axioms,frongillo2023axiomatic}.

Within this work, we do not attempt to compare the classical centralized structure of running a prediction market with the currently proposed decentralized framework. We focus instead on providing the key properties of such a market -- the ability to provide a trusted counterparty throughout the entire life of the market and collect fees for performing that service -- while also allowing investors to pool resources and act as passive investors in the market.

\subsection{Primary contributions}\label{sec:intro-contributions}

As the goal of this work is to construct decentralized prediction markets via AMMs, our primary contributions are threefold.
\begin{itemize}
\item We generalize the classical prediction market AMM setting of \cite{hanson2007logarithmic} from a finite to a general probability space. In particular, we follow the approach of~\cite{chen2012utility,othman2013practical} to study these AMMs as utility functions with an \textbf{axiomatic construction}. In doing so, we investigate the properties of the cost function and pricing oracle of these AMMs. For instance, we demonstrate that with our proposed structure, a pricing measure is guaranteed to exist between the bid and ask pricing oracles.
\item We formalize the \textbf{decentralized liquidity provision} for prediction market AMMs to demonstrate that these markets can be made in a decentralized framework. We provide details on how to add or withdraw liquidity both before the market opens (as is standard in the classical framework) and 
 after some bets have already been placed. Additionally, we investigate the impacts that these variations in liquidity have on market behavior. Notably, this added liquidity is \emph{not} deterministic but rather a bet itself. To the best of the authors' knowledge, no prior work on prediction markets has allowed for adding liquidity after the market opens. 
\item We define a novel \textbf{fee structure} to retain the important financial properties on the cost function while allowing liquidity providers to collect a profit from their investments. 
Prior works have considered fee structures for AMMs via, e.g., arbitrage opportunities or market trading frequency/volatility \cite{milionis2022automated,milionis2023automated,cao2023automated}; however, these approaches are not well suited for a prediction AMM.
As far as the authors are aware, no prior work on prediction markets has \emph{explicitly} defined the fees to be collected.\footnote{Other works (e.g., \cite{ostrovsky2009information}) consider an implicit bid-ask spread. The implicit construction differs from the explicit fees as they do not guarantee that the market maker earns a riskless profit from transacting.}
\end{itemize}

\subsection{Organization of this Paper}\label{sec:intro-organization}
The organization of this paper is as follows.  
A summary and generalization of prior works on prediction markets is provided within Section~\ref{sec:hanson}. This discussion is followed by our proposed structure for a prediction market AMM in an axiomatic framework within Section~\ref{sec:amm-construction}. With that discussion, we formulate meaningful mathematical and financial properties for the cost function and pricing oracle in Section~\ref{sec:amm-properties}. For instance, we demonstrate that there natively exists bid and ask pricing oracles with a pricing probability measure on the underlying measurable space sandwiched between these prices.  In Section~\ref{sec:pooling}, we study the problem of decentralized liquidity provision for these prediction markets. In particular, we discuss how to provide liquidity so as to decrease price slippage during trading. In Section~\ref{sec:fees}, we present a discussion of how to explicitly collect fees on bets in this market. Notably, due to the general construction utilized herein, these fees are assessed in a non-trivial manner so as to retain the important financial properties of the AMM.
Finally, in Section~\ref{sec:cases}, we consider two numerical case studies. First, we provide an empirical study of the behavior of an AMM when applied to data on sports betting for the 2023 Super Bowl. Second, we apply a prediction market AMMs for the purposes of quoting financial options prices.
We summarize and conclude in Section~\ref{sec:conclude}.

\section{Generalized Structure for Automated Market Makers}\label{sec:hanson}

Fundamentally, an AMM is a liquidity pool -- made up of deposits from liquidity providers -- against which traders can execute transactions at prices based on mathematical algorithms without requiring a human-in-the-loop.
We take the view from DeFi that AMM algorithms are equivalent to utility functions (see, e.g., \cite{schlegel2022axioms,bichuch2022axioms}); the relationship between AMMs and utility functions was previously explored for prediction markets within~\cite{chen2007utility}.
The primary purpose of this manuscript is to explain how prediction market AMMs can be constructed so as to explicitly consider the actions, as well as the risks and rewards, of liquidity providers to guarantee there is liquidity available to cover the costs of the winning bets.
The AMM begins accepting bets once it has some amount of liquidity and allows bettors to place bets up until some fixed time or event, at which time the winning bets are paid out.\footnote{Herein, we do not explicitly consider the evolution of information as would happen, e.g., for in-game sports betting.} In Section~\ref{sec:pooling}, we consider how to allow for new liquidity provision (or withdrawals of liquidity by liquidity providers) even after the market has already begun accepting bets. However, since the payoffs of the winning bets come from the liquidity pool, liquidity provision is an inherently risky position; because of this, liquidity providers need to be compensated with fees (as disucssed in Section~\ref{sec:fees}). While bettors have a payoff based on the specific bet placed, the liquidity providers share the fees and any liquidity remaining after all bets are paid off.

We wish to begin our study of automated market makers for prediction markets by recalling (a generalization of) Hanson's market maker~\cite{hanson2007logarithmic}. This construction follows the utility function framework of \cite{chen2007utility} but with explicit consideration for the cash reserves available.

\begin{definition}\label{defn:general}
Consider a probability space $(\Omega,\F,\P)$.\footnote{In this context, $\P$ represents a physical measure that is not used throughout our analysis, except in examples.}
Let $L^\infty$ denote the space of uniformly bounded random variables and $L^\infty_+ := \{x \in L^\infty \; | \; x \geq 0~\text{a.s.}\}$.  An \textbf{\emph{automated market maker (AMM)}} is a utility function $U: L^\infty_+ \times \R_+ \to \R \cup \{-\infty\}$ that is non-increasing in its first (random-valued) component and non-decreasing in its second (real-valued) component.
\end{definition}
We can interpret an AMM $U: L^\infty_+ \times \R_+ \to \R \cup \{-\infty\}$ as the utility $U(\pi,L)$ of betting positions $\pi \in L^\infty_+$ and cash reserves $L \geq 0$.  More specifically, $\pi(\omega)$ denotes the payout that would need to be made by the market maker in the event that outcome $\omega \in \Omega$ is realized; $L \geq 0$ denotes the capital held by the market maker that is available to be used for the eventual payouts when the predicted event is realized. Bets with negative payouts, e.g., from selling a bet to the AMM, can be accommodated by buying the complement and shifting the cash reserves accordingly as will be made clear in~\eqref{eq:CI} below.

The classical goal of a prediction market is to be able to give a price to a new incoming bet $x\in L^\infty$. This, of course, depends on the current state of the market $(\pi, L)$. As opposed to a standard prediction market, where the cost function of the bet plays the key role, here we utilize the utility indifference approach from DeFi (\cite{schlegel2022axioms,bichuch2022axioms}). In this way, the AMM must not be worse off after a new bet is placed (and the cost of the bet was collected) than it was before this bet was placed. This in turn allows us to define the cost function of any bet. We proceed as follows. That is, to place a bet with a payout of $x \in L^\infty$, the bettor must make a payment of
\begin{equation}\label{eq:CI}
C(x;\pi,L) = \inf\{c \in [\essinf x,\esssup x] \; | \; U(\pi + x - \essinf x \vec{1} , L+c-\essinf x) \geq U(\pi,L)\}, 
\end{equation}
with initial betting positions $\pi \in L^\infty_+$ and cash reserves $L \geq 0$.\footnote{Prior works on prediction market AMMs, e.g.~\cite{chen2007utility}, formulate the cost via integration of infinitely small transactions instead as they do not explicitly consider the available cash reserves $L$; it can readily be shown that these formulations are equivalent. In doing so, prior works focus on proving there exists a bounded loss so that there is some minimal liquidity that guarantees AMM solvency.} Throughout this work we define $\vec{1} \in L^\infty$ to be equal to $1$ almost surely (a.s.).
This is the minimal amount of additional cash necessary to guarantee that the market maker has a nondecreasing utility. Explicitly, in this way, we separate the random part of the bet ($x-\essinf x \vec{1}$) from its direct impact on the cash reserves ($-\essinf x$). In fact, this formulation encodes that a trader provides her own cash reserves to settle the constant part ($\essinf x \vec{1}$) of any bet.
We wish to remark that the setting in which $x \geq 0$ almost surely corresponds to purchasing bets whereas $\P(x < 0) > 0$ allows for the trader to sell a bet to the AMM (i.e., receive a guaranteed payment today in exchange for paying if a specific event occurs). Though we do not impose it here, it may be desirable for the AMM to impose a no-short selling constraint on traders so that the AMM does not need to worry about counterparty risk in the future.

\begin{remark}\label{rem:indifference}
The purchasing price $C(x;\pi,L)$ for some bet $x \in L^\infty$ can be viewed as a utility indifference price (see e.g. \cite{carmona2008indifference}). Specifically, assuming $U$ is strictly increasing in its second (cash reserves) input and sufficiently continuous, we can rewrite $c = C(x;\pi,L)$ as the unique solution to $U(\pi+x-\essinf x \vec{1};L+c-\essinf x) = u(\pi,L)$ provided it exists (i.e., the AMM has sufficient cash reserves to trade the bet $x$). Herein, we choose to follow the optimization formulation of~\eqref{eq:CI} as existence and uniqueness of the cost are trivial for all $(x;\pi,L) \in L^\infty \times L^\infty_+ \times \R_+$.
\end{remark}

In addition to acting as a liquidity provider, the AMM can also provide \emph{pricing oracles} $P^a, P^b: L^\infty \times L^\infty_+ \times \R_+ \to [0,1]$ which give the marginal cost of placing positive and negative bets respectively. More specifically, we define the ask and bid pricing oracles, respectively, as the cost of placing positive and negative marginal sized bets:
\begin{equation}\label{eq:price1}
P^a(x;\pi,L) := \lim_{t \searrow 0} \frac{1}{t} C(tx;\pi,L)\ \ \text{and} \ \ P^b(x;\pi,L) := \lim_{t \nearrow 0} \frac{1}{t} C(tx;\pi,L),
\end{equation}
for any bet $x \in L^\infty$, betting position $\pi \in L^\infty_+$ and cash reserves $L \geq 0$. It is desirable that $P^b(\cdot,;\pi,L) \le P^a(\cdot,;\pi,L)$, and that there exists a probability measure $\Q$ such that the bid and ask prices sandwich an expectation with respect to $\Q$. In general, this sandwiched probability measure $\Q$ will depend on the current state $(\pi, L)$.

Assuming that the cost function has sufficient mathematical regularity to guarantee its differentiability, the bid and ask spread disappears, and the (unique) price becomes
\begin{equation}\label{eq:P}
P(x;\pi,L) := \frac{\d}{\d t} C(t x;\pi,L)|_{t = 0}.
\end{equation}
This also selects the unique measure $\Q$ (for the fixed $(\pi,L)$), the expectation with respect to which must then be equal to the pricing oracle $P$, i.e., $\Q(A) := P(\ind_{A};\pi,L)$ for any $A \in \F$. As such, we can view $\Q$ as a pricing measure, providing a `quoted price' for $x \in L^\infty$ via $\E^\Q[x]$ under market conditions $(\pi,L)$. Notably, this quoted price differs from $C(x;\pi,L)$ due to the price slippage inherent in the form of $C$ given in \eqref{eq:CI}.

\begin{example}\label{ex:hanson}
Within this example we want to study the classical structure of Hanson's AMM~\cite{hanson2007logarithmic} with logarithmic scoring rule.   
As provided in~\cite{chen2007utility}, this AMM can be formulated according to Definition~\ref{defn:general} as a utility function $U^H: L^\infty \times \R_+ \to \R \cup \{-\infty\}$ for any bet sizes $\pi \in L^\infty_+$ and cash reserves $L \geq 0$ as
\[U^H(\pi,L) := \log\left(\E\left[1-\exp(-\gamma(L\vec{1}-\pi))\right]\right)\]
for $\gamma > 0$ where, by convention, $\log(x) = -\infty$ for any $x \leq 0$.
By construction of this utility function, the cost of purchasing a bet with payout of $x \in L^\infty$ is explicitly provided by:
\[C^H(x;\pi,L) = \frac{1}{\gamma}\log\left(\frac{\E[\exp(\gamma(\pi + x))]}{\E[\exp(\gamma\pi)]}\right).\]
In this case $\frac{d\Q}{d\P} = \frac{\e{\gamma \pi}}{\E[\e{\gamma\pi}]}$ is the pricing measure under market conditions $(\pi,L)$; notably, while $\Q$ does \emph{not} depend on the existing reserves $L$, it does depend on the existing bets $\pi$.
Recall from Remark~\ref{rem:indifference} that the cost $C^H$ can be viewed as a utility indifference price. Specifically writing $U(\pi,L) = u(L\vec{1}-\pi)$, and assuming $u$ is strictly increasing, we can rewrite $c = C^H(x;\pi,L)$ as the unique solution to $u(L\vec{1}+c\vec{1}-\pi-x) = u(L\vec{1}-\pi)$.

Often, prediction markets are only considered with \emph{finite} probability spaces \cite{hanson2007logarithmic}.  Here, and in other examples below, we will utilize the finite probability space $(\Omega_N,2^{\Omega_N},\P_N)$ with $\P_N(\omega) = \frac{1}{N}$ for every $\omega \in \Omega_N$ and $\operatorname{card}(\Omega_N) = N$, for some $N \in \mathbb{N}$.
Hanson's AMM is more commonly written as the integral of the marginal costs~\cite{berg2009hanson}, i.e., for $x \in \R$,
\[C^H(x\ind_{\bar\omega};\pi,L) = \int_0^x \frac{\exp(\gamma(\pi(\bar\omega)+z))}{\exp(\gamma(\pi(\bar\omega)+z)) + \sum_{\omega \neq \bar\omega} \exp(\gamma\pi(\omega))} dz.\]
As is made clear from the cost function $C^H$, Hanson's AMM prices bets independent of the cash reserves available.  However, in practice, it is important that the market maker has sufficient cash reserves to cover any bets that need to be paid out, i.e., $L \geq \esssup\pi$ at any realized bets $\pi \in L^\infty_+$ and cash reserves $L > 0$.
Formally, we want $L + C^H(x;\pi,L) \geq \esssup(\pi+x)$ for any $x\in L^\infty$, $\pi \in L^\infty_+$ and $L > \esssup \pi$ so that the cash reserves after a bet is placed is always sufficient to cover the worst case payouts (and assuming such a property holds before $x$ is traded). By bounding the log-sum-exponential functions, it can be readily shown that $L + C^H(x;\pi,L) > L - \log(N)/\gamma - \esssup \pi + \esssup(\pi+x)$. That is, sufficient cash reserves exist if $L \geq \log(N)/\gamma + \esssup \pi$ or, equivalently, the risk aversion $\gamma \geq \log(N)/[L - \esssup \pi]$ is sufficiently large. Beginning from $\pi = \vec{0}$ and initial cash reserves $L_0 > 0$, this gives an initial bound $\gamma \geq \log(N)/L_0$; in fact, in this initial time point, this bound is necessary and sufficient for the AMM to have the requisite cash reserves for any bet. Furthermore, this inital bounding condition $\gamma \geq \log(N)/L_0$ is sufficient for the AMM to be liquid throughout its operation.\footnote{This can be demonstrated due to the path independence property as introduced in Theorem~\ref{thm:properties}\eqref{split} below.}
\end{example}

The generalized structure of AMMs presented here is broad so as to encompass even structures that may not have sufficient cash reserves to cover all possible bets (i.e., if $L < \esssup\pi$). 
For instance, if $\gamma < \log(N)/L_0$ for Hanson's AMM presented in Example~\ref{ex:hanson} over a finite probability space, then there exists some bet that creates a probability that the market maker will default on (a fraction of) its obligations; notably this occurs once an infinite number of events needs to be considered.
Taking the inspiration from the form $U^H$ and following the structure of the utility functions in~\cite{chen2007utility}, for the remainder of this work, we will focus on a special structure for AMMs based on the remaining cash reserves for all outcomes, i.e., such that $U(\pi,L) := u(L\vec{1}-\pi)$ for some continuous and nondecreasing utility function $u: L^\infty_+ \to \R \cup \{-\infty\}$.  As will be shown in the subsequent sections, these cash reserves-based AMMs can readily guarantee sufficient cash reserves will always exist (as well as other useful properties for a market); we refer the interested reader to~\cite{chen2007utility} which discusses the maximum loss for such an AMM under the finite probability space $(\Omega_N,2^{\Omega_N},\P_N)$.  
Additionally, by specifically accounting for cash reserves in the structure of an AMM, we allow for changes and dynamics in cash reserves to occur. In particular, we permit investors to pool cash into the AMM (even after betting has started) to collect a portion of the market proceeds; this notion of pooling is expanded upon in Section~\ref{sec:pooling}. This decentralization can lead to more responsive liquidity provision as investors flock to volatile markets (to collect fees, see Section~\ref{sec:fees}). Due to monotonicity (and as will be formally proven below), this increased liquidity has the ancillary benefit of reducing the price impacts from trading which ex-ante drives down the volatility. 

\section{Liquidity-Based Automated Market Makers}\label{sec:amm}
As discussed above, within this section we wish to study those automated market makers of the form $U(\pi,L) := u(L\vec{1}-\pi)$ for some utility function $u: L^\infty_+ \to \R \cup \{-\infty\}$.\footnote{In the subsequent sections we will explicitly provide the domain $\dom u := \{x \in L^\infty_+ \; | \; u(x) > -\infty\}$ of $u$.}  
As introduced in Section~\ref{sec:hanson}, we again define a prediction market AMM by its utility function $u$ rather than a cost function of a bet $C$; we then derive the cost function using the utility indifference argument.
Within Section~\ref{sec:amm-construction}, we propose some basic axioms that any such liquidity-based automated market maker should satisfy.  The implications of these axioms on bets are investigated within Section~\ref{sec:amm-properties}.
To simplify notation, for the remainder of this work we will let $\Pi := L\vec{1}-\pi$ denote the liquidity remaining for each outcome (contingent on that outcome being realized); we will demonstrate in Section~\ref{sec:amm-properties} that we can always guarantee that $\essinf \Pi > 0$ under the desired axioms (provided that $L_0 > 0$ for the initial cash reserves of the market maker).

\subsection{Construction}\label{sec:amm-construction}
As expressed in the introduction of this section, our first goal is to characterize the liquidity-based automated market makers as utility functions $u: \domu := \{x\in  L^\infty_+ \; | \; \essinf x>0\} \to \R$ satisfying useful mathematical properties.\footnote{We will set $u(x) = -\infty$ for $x \in L^\infty\backslash\domu$.} The following definition encodes the minimal set of axioms that we utilize throughout this work.
\begin{definition}\label{defn:amm}
A utility function $u: \domu := \{x\in  L^\infty_+ \; | \; \essinf x>0\} \to \R$ is a \textbf{\emph{liquidity-based automated market maker [LBAMM]}} if:
\begin{enumerate}
\item\label{defn:amm-1} $u$ is upper semi-continuous on $L^\infty_+$ with respect to the weak* topology;\footnote{For any $x$ in the domain of $u$ and any $\epsilon > 0$, there exists a neighborhood of $x$ in the weak* topology on $L^\infty$ such that for all $y$ in this neighborhood, $u(y) \leq u(x) + \epsilon$.}
\item\label{defn:amm-2} $u$ is strictly increasing;\footnote{For any $x_1, x_2 \in \domu$, if $x_1 \geq x_2$ a.s.\ and $\P(x_1 > x_2) > 0$, then $u(x_1) > u(x_2)$.} and
\item\label{defn:amm-3} $u$ is quasiconcave.\footnote{For all $x_1,x_2\in \domu$ and for all $t\in [0,1]$, $u(tx_1+(1-t)x_2)\geq \min\{u(x_1), u(x_2)\}$.} 
\end{enumerate}
\end{definition}
\begin{remark}\label{rk-cvg}
Note that from Definition \ref{defn:amm}\eqref{defn:amm-1}-\eqref{defn:amm-2}, it follows (see, e.g.,~\cite[Theorem 4.31]{follmer2008stochastic}) that $u$ is almost surely continuous from above, i.e., if $x_n \searrow x$ a.s.\ then $u(x_n) \searrow u(x)$.
\end{remark}

Let $\Pi \in \domu$ denote the liquidity remaining for each outcome. Then as described in~\eqref{eq:CI} above, the cost under an LBAMM of purchasing a payoff of $x \in L^\infty$ is defined such that
\begin{equation}\label{eq:C}
C(x;\Pi) := \inf\{c \in [\essinf x,\esssup x] \; | \; u(\Pi - x + c\vec{1}) \geq u(\Pi)\}.
\end{equation}
To simplify notation, where clear from context, we will drop the explicit dependence of $C$ on $\Pi$, i.e., we will denote $C(x) := C(x;\Pi)$.

Before continuing our study of the properties of LBAMMs, we provide a simple example of one such AMM in a finite probability space.  This construction is based on Uniswap V2 which is a popular AMM in decentralized finance~\cite{uniswapv2}.  As we will see, this construction provides a closed form representation for the cost of a bet in the 2 event setting. Notably, this AMM was first discussed in~\cite{chen2007utility} for prediction markets (in a finite probability space) before its use in decentralized finance. 
\begin{example}\label{ex:uniswap}
Consider the finite probability space $(\Omega_N,2^{\Omega_N},\P_N)$ as in Example \ref{ex:hanson} with $\P_N(\omega) = \frac{1}{N}$ for every $\omega \in \Omega_N := \{\omega_1,\dots, \omega_N\}$ with cardinality $N$. Consider the logarithmic utility function $u(\Pi) = \E[\log(\Pi)] = \frac{1}{N} \sum_{i = 1}^N \log(\Pi(\omega_i))$ if $\min\Pi > 0$ and $u(\Pi) = -\infty$ otherwise. As discussed in, e.g., \cite{bichuch2022axioms}, the constant product market maker of Uniswap V2 is equivalent to a logarithmic utility indifference pricing; for this reason we consider this LBAMM to be the prediction market version of Uniswap V2. Then, for any $x \in L^\infty$, the cost for the bettor 
$$C(x; \Pi)=\inf\left\{c\in [\min x,\max x] \; | \; \prod_{i=1}^N (\Pi(\omega_i)-x(\omega_i)+c)\geq \prod_{i=1}^N \Pi(\omega_i)\right\}$$ satisfies a constant product rule
\[\prod_{i = 1}^N (\Pi(\omega_i)-x(\omega_i)+C(x;\Pi)) = \prod_{i = 1}^N \Pi(\omega_i).\]
In particular, for the setting with only $N=2$ possible outcomes,
\begin{align}
C(x; \Pi)&=
\frac{-\Pi(\omega_1)-\Pi(\omega_2)+x(\omega_1)+x(\omega_2)}{2}\\
&+\frac{\sqrt{\bigl[(\Pi(\omega_1)-x(\omega_1))-(\Pi(\omega_2)-x(\omega_2))\bigr]^2+4\Pi(\omega_1)\Pi(\omega_2)}}{2}
\end{align}
for any $x \in L^\infty$ and $\Pi \in \domu$.
We implement this utility function within Section~\ref{sec:superbowl} as an empirical case study.
\end{example}

We wish to note that the logarithmic utility of Example~\ref{ex:uniswap} in a general probability space does not, necessarily, satisfy all conditions of an LBAMM. In the following example, which concludes this discussion of the definition of LBAMMs, we introduce a modification which can be used to extend the logarithmic utility to general probability spaces. 

\begin{example}\label{ex:essinf}
Consider a general probability space $(\Omega,\F,\P)$. Let $u:L_+^\infty\to \R\cup\{-\infty\}$ be a strictly increasing, concave, and weak* upper semi-continuous function with $\dom u \supseteq \domu$. (For example, one might choose the logarithmic utility function $u(x) = \E[\log(x)]$ as seen in Example~\ref{ex:uniswap} or, more generally, set $u(x) = \E[\log(x)] + \lambda \log(\E[x])$ for some $\lambda \geq 0$ as taken in the Liquid StableSwap of~\cite{bichuch2022axioms}.) Fix $\epsilon\in (0,1)$, and define $\bar{u}_\epsilon: L_+^\infty \to \R\cup\{-\infty\}$ such that for all $x\in L_+^\infty$: 
\[\bar{u}_\epsilon(x) := (1-[\epsilon - \inf_{A \in \F_+} \P(A)]^+)u(x) + [\epsilon - \inf_{A \in \F_+} \P(A)]^+ \log(\essinf x),\] where $\F_+:=\{A\in \F \mid \P(A)>0\}$. It is straightforward to verify that this utility function satisfies all the required properties of an LBAMM as per Definition~\ref{defn:amm}. 
Note that, under the finite probability space $(\Omega_N,2^{\Omega_N},\P_N)$, if $\epsilon \leq \frac{1}{N}$ then $\bar{u}_\epsilon \equiv u$.
This utility function structure is implemented within Section~\ref{sec:derivative} to construct a simple financial derivatives market.\footnote{The gas fees for using a general utility function $\bar{u}_\epsilon$ can be controlled by utilizing modern blockchains such as Avalanche and Skale which charge developers a monthly fee for gas-less transactions.}
\end{example}

\subsection{Properties}\label{sec:amm-properties}
Given the construction of an LBAMM in Definition~\ref{defn:amm}, we can now consider the formal properties of the cost functions $C$ defined in~\eqref{eq:C}.

\begin{theorem}\label{thm:properties}
Let $u: \domu \to \R$ be an LBAMM with remaining liquidity $\Pi \in \domu$ and associated payment function $C: L^\infty \to \R$.
\begin{enumerate}
\item\label{bounds} \textbf{No arbitrage:} $C(x) \in [\essinf x,\esssup x]$ for $x \in L^\infty$ with attainment $C(x) \in \{\essinf x,$ $\esssup x\}$ if and only if $x = c\vec{1}$ for some $c \in \R$.
\item\label{equal} \textbf{Liquidity-bounded loss:} $u(\Pi-x+C(x)\vec{1}) = u(\Pi)$ and $\Pi-x+C(x)\vec{1} \in \domu$ for any $x \in L^\infty$. 
\item\label{convex-monotone} \textbf{Convex, monotone and lower semi-continuous:} $x \mapsto C(x)$ is strictly increasing, convex and lower semi-continuous (in the weak* topology). 
\item\label{split} \textbf{Path independent:} $C(x+y;\Pi) = C(x;\Pi) + C(y;\Pi-x+C(x;\Pi)\vec{1})$ for $x,y \in L^\infty$. As a direct consequence $C(x+c\vec{1};\Pi) = C(x;\Pi) + c$ for any $x \in L^\infty$ and $c \in \R$.
\end{enumerate}
\end{theorem}
\begin{proof}
\begin{enumerate}
\item By construction of the cost function $C$, it easily follows that $C(x) \in [\essinf x,\esssup x]$. If $x = c\vec{1}$ for some $c \in \R$, then trivially, $C(x) = c = \essinf x = \esssup x$. 
Now, let's consider the case where $C(x) = \essinf x$ (the case of $C(x) = \esssup x$ follows similarly). If $x \in L^\infty\backslash\{c\vec{1}|c\in\R\}$, then by the strict monotonicity of the utility function (and using the fact that $\Pi-x+C(x)\vec{1} \in \domu$ from Property~\ref{equal}), we have $u(\Pi-x+C(x)\vec{1}) = u(\Pi-[x-\essinf x\vec{1}]) < u(\Pi)$. This contradicts Property~\ref{equal}, which is proved below.
 
 \item Note that the sequence $\bigl(\Pi-x+(C(x)+\frac 1 n)\vec{1}\bigr) \searrow \Pi-x+C(x)\vec{1}$ a.s.\ as $n\to \infty$. Then, by Remark~\ref{rk-cvg}, we have
$u(\Pi-x+C(x)\vec{1})=\lim_{n\to \infty} u\bigl(\Pi-x+(C(x)+\frac 1 n)\vec{1}\bigr) \geq u(\Pi),$
as implied by the construction of $C$. Consequently, $\essinf(\Pi-x+C(x)\vec{1})>0$ which ensures that $\Pi-x+C(x)\vec{1} \in \domu$ for all $x\in L^\infty$. Furthermore, given that $u(\Pi-x+c^*\vec{1})<u(\Pi)$ for $c^*=\max\{\essinf(x), -\essinf(\Pi-x)\}$, the equality $u(\Pi-x+C(x)\vec{1}) = u(\Pi)$ must be satisfied due to the property of continuity from above.
 
 \item
\begin{enumerate}
\item \underline{Monotonicity:} To prove the monotonicity of $C$, let's consider an arbitrary $x \in L^\infty$ and $\delta \in L^\infty_+\backslash \{\vec{0}\}$. Assume $C(x+\delta) \leq C(x)$. By the strict monotonicity of the utility function and 
Property~\eqref{equal}, we arrive at a contradiction:
    \[u(\Pi) = u(\Pi-[x+\delta]+C(x+\delta)\vec{1}) < u(\Pi-x+C(x)\vec{1}) = u(\Pi).\]
    
    \item \underline{Convexity:} To prove the convexity of $C$ we will consider its epigraph.  That is, 
    \begin{align}
    \operatorname{epi}C = \{(x,c) \in L^\infty \times \R \; | \; C(x) \leq c\} = \{(x,c) \in L^\infty \; | \; u(\Pi-x+c\vec{1}) \geq u(\Pi)\}.~~~~~~~
\label{eq:epi-C}
    \end{align}
Since $u(\Pi) \in \R$ is a constant, the epigraph of $C$ is convex due to the quasiconcavity of the utility function $u$.

\item \underline{Lower semi-continuity:} As seen in \eqref{eq:epi-C}, the epigraph of $C$ corresponds to the hypograph for $u$. Therefore, by upper semi-continuity of $u$, the function $C$ is lower semi-continuous.

\end{enumerate}

\item Let $x,y \in L^\infty$. By~\eqref{equal}, we can immediately conclude
\[u(\Pi-x+C(x;\Pi)\vec{1}-y+C(y;\Pi-x+C(x;\Pi)\vec{1} )\vec{1}) = u(\Pi-x+C(x;\Pi)\vec{1}) = u(\Pi).\]
By construction of $C(x+y;\Pi)$ and the strict monotonicity of $u$, it must hold that $C(x+y;\Pi) = C(x;\Pi) + C(y;\Pi-x+C(x;\Pi)\vec{1})$.
The second property provided on the translativity of $C$, i.e., $C(x+c\vec{1};\Pi) = C(x;\Pi) + c$ for any $c \in \R$ is a direct consequence by recalling from property~\eqref{bounds} that $C(c\vec{1};\Pi-x+C(x;\Pi)\vec{1}) = c$.

\end{enumerate}
\end{proof}

\begin{remark}
\begin{enumerate}
\item Theorem~\ref{thm:properties}\eqref{equal} directly connects the LBAMM pricing scheme to a utility indifference price (see, e.g.,~\cite{carmona2008indifference}) as previously discussed in Remark~\ref{rem:indifference}. 
Additionally, we can interpret the LBAMM cost function $C$ as a nonlinear expectation due to Theorem~\ref{thm:properties}\eqref{bounds} and \eqref{convex-monotone}.
\item The second property of Theorem~\ref{thm:properties}\eqref{equal} implies there is infinite liquidity and infinite sized bets can be placed.  This guarantees that the AMM will never default because, under almost every realized outcome $\omega \in \Omega$, there still remains $\Pi(\omega)-x(\omega)+C(x) > 0$ cash after paying off all bets.
\item As a direct consequence of Theorem~\ref{thm:properties}\eqref{bounds} and~\eqref{split}, a clear no round-trip arbitrage argument follows.  Specifically, for any $x \in L^\infty$ and $c \in \R$, \[c = C(c\vec{1};\Pi) = C(x;\Pi) + C(c\vec{1}-x;\Pi-x+C(x;\Pi)\vec{1}).\]  That is, a guaranteed payoff of $c$ can only be obtained at a cost of $c$.  In particular, $C(x;\Pi) + C(-x;\Pi-x+C(x;\Pi)\vec{1}) = 0$ so that buying, and immediately selling, a bet results in no profits (or losses) for the trader.  This can be viewed as a version of the prior notion that the LBAMM can never default as, if an arbitrage such as this existed, traders could exploit this design flaw in order to guarantee profits at the expense of the LBAMM, which in the extreme case can cause the LBAMM to default on payments.
\end{enumerate}
\end{remark}

\begin{corollary}\label{cor:lipschitz}
Let $u: \domu \to \R$ be an LBAMM with remaining liquidity $\Pi \in \domu$. The associated payment function $C: L^\infty \to \R$ is Lipschitz continuous in the bet size with Lipschitz constant $1$.
\end{corollary}
\begin{proof}
Let $x,y \in L^\infty$. By the monotonicity and translativity proven in Theorem~\ref{thm:properties}:
\[C(x)-C(y) \leq C(y + \|x-y\|_\infty \vec{1}) - C(y) = \|x-y\|_\infty.\]
By symmetry between $x$ and $y$, Lipschitz continuity follows.
\end{proof}

Let $\Pi \in \domu$ be fixed. We now wish to consider the ask and bid pricing oracles, $P^a,P^b: L^\infty \to \R$, associated with the liquidity-based AMMs. Specifically, as constructed for the general AMMs in Section~\ref{sec:hanson}, we define the ask and bid pricing oracles, respectively, as the cost of placing a positive and negative marginal sized bet. That is,
\begin{equation}\label{eq:price}
P^a(x) := \lim_{t \searrow 0} \frac{1}{t} C(tx) \ \ \text{and} \ \ P^b(x) := \lim_{t \nearrow 0} \frac{1}{t} C(tx)=-P^a(-x).
\end{equation}

\begin{theorem}\label{thm:prices}
Let $u: \domu \to \R$ be an LBAMM with remaining liquidity $\Pi \in \domu$. The associated pricing oracles $P^a, P^b: L^\infty \to \R$, defined in~\eqref{eq:price}, satisfy $P^a(x)\geq P^b(x)$ for all $x\in L^\infty$. Furthermore, there exists a measure $\Q\sim\P$  such that $P^b(x)\leq \E^\Q[x]\leq P^a(x)$ for all $x\in L^\infty$.
\end{theorem}
\begin{proof}
Since $C$ is convex by Theorem~\ref{thm:properties}, it is straightforward to verify that $P^a, P^b: L^\infty \to \R$ are well-defined and $\rho(x):=P^a(-x)$ is a lower semicontinuous coherent risk measure. 
In fact, following the nomenclature of \cite{delbaen2002coherent}, $\rho$ is also relevant as $P^a(\ind_A) > 0$ for every $A \in \F_+$ (i.e., $A \in \F$ with $\P(A) > 0$).\footnote{Trivially $P^a(\ind_A) = 1 - P^b(\ind_{A^c}) > 0$ for any $A \in \F_+$ as $P^b(\ind_{A^c}) \leq P^a(\ind_{A^c}) \leq C(\ind_{A^c}) < 1$ by no arbitrage and convexity (Theorem~\ref{thm:properties}).}
By the dual representation of coherent risk measures (see, e.g.,~\cite[Corollary 4.34]{follmer2008stochastic}) , there exists a set of probability measures $\mathcal{Q} \ne \emptyset$ which are equivalent to 
 $\P$, i.e., $\mathcal{Q}\subseteq \{\Q\sim\P\}$, such that
$P^a(x) = \sup \bigl\{\E^\Q[x]\;|\;\Q\in \mathcal{Q}\bigr\}$. Moreover, by construction, $P^b(x)= \inf \bigl\{\E^\Q[x]\;|\;\Q\in \mathcal{Q}\bigr\}$. We thus conclude that $P^a(x)\geq \E^\Q[x]\geq P^b(x)$ for all $x\in L^\infty$ and all $\Q\in \mathcal{Q}$.
\end{proof}

\begin{remark}
From Theorem~\ref{thm:prices}, we note that:
\begin{enumerate} 
\item Theorem~\ref{thm:prices} guarantees the existence of an equivalent \emph{pricing measure} $\Q \sim \P$ that sits within the bid-ask spread, i.e., $P^b(x) \leq \E^\Q[x] \leq P^a(x)$ for any $x \in L^\infty$. We remind the reader that this equivalent pricing measure $\Q$ implicitly depends on the remaining liquidity $\Pi \in \domu$. We regard any such measure as a pricing measure as it is consistent with the quoted (bid and ask) prices in the market. 
\item The pricing oracle $P^a$ is, in fact, the largest coherent risk measure that is dominated by the cost function $C$. 
A similar statement can be given for the pricing oracle $P^b$ on the negative of a bet.
\item If $C$ is differentiable at 0, then there is no bid-ask spread, and $P^a(x)=P^b(x)$ for every $x \in L^\infty$. This differentiability follows if $u$ is Fr\'echet differentiable; however, we note that the construction in Example~\ref{ex:essinf} is \emph{not} Fr\'echet differentiable if $\epsilon > \inf_{A \in \F_+} \P(A)$.
\end{enumerate}
\end{remark}

We conclude our discussion of the properties of the LBAMM by considering the outcome of a bet from an infinitely liquid and risk-neutral bettor with beliefs $\Q \ll \P$, i.e., who is solving
\begin{equation}\label{eq:max-bet}
\sup_{x \in L^\infty} \E^\Q[x - C(x;\Pi)\vec{1}].
\end{equation}
Such a result serves as a converse to Theorem~\ref{thm:prices} as this bettor has a (finite) optimal bet $x^* \in L^\infty$ if, and only if, $P^b(x;\Pi - x^* + C(x^*;\Pi)\vec{1}) \leq \E^\Q[x] \leq P^a(x;\Pi - x^* + C(x^*;\Pi)\vec{1})$ for any $x \in L^\infty$ so that $\Q$ is sandwiched by $P^b,P^a$.
Fundamentally, this sandwich property implies an LBAMM is able to extract the subjective beliefs (i.e., $\Q$) of any risk-neutral bettor. Such a property was first considered within \cite{hanson2007logarithmic} in which the Hanson's logarithmic scoring rule in a finite probability space was proven to ``extract the information implicit in the trades others make with it, in order to infer new rational prices.''
\begin{theorem}\label{thm:max-bet}
Let $u: \domu \to \R$ be an LBAMM with remaining liquidity $\Pi \in \domu$. 
Consider a risk-neutral bettor with subjective measure $\Q \ll \P$, optimizing~\eqref{eq:max-bet}.
For any $x^* \in \argmax_{x \in L^\infty} \E^\Q[x - C(x;\Pi)\vec{1}]$, it holds that $$P^b(x;\Pi - x^* + C(x^*;\Pi)\vec{1}) \leq \E^\Q[x] \leq P^a(x;\Pi - x^* + C(x^*;\Pi)\vec{1}).$$
Furthermore, if $Q \in \M_* := \{\Q \ll \P \; | \; \essinf \frac{d\Q}{d\P} > 0\}$ then $\argmax_{x \in L^\infty} \E^\Q[x - C(x;\Pi)\vec{1}] \neq \emptyset$. 
\end{theorem}
\begin{proof}
By definition of the pricing oracles $P^a,P^b$ as directional derivatives, the sandwich property holds if $\frac{d\Q}{d\P} \in \partial C(\vec{0};\Pi-x^*+C(x^*;\Pi)\vec{1})$ for $x^* \in \argmax_{x \in L^\infty} \E^\Q[x - C(x;\Pi)\vec{1}]$ where $\partial C$ denotes the subdifferential of $C$.
By construction of $x^*$ as a maximizer, $\E^\Q[x^* - C(x^*;\Pi)\vec{1}] \geq \E^\Q[x - C(x;\Pi)\vec{1}]$ for any $x \in L^\infty$. Rearranging terms, we recover $C(x;\Pi) - C(x^*;\Pi) \geq \E^\Q[x - x^*]$ for any $x \in L^\infty$, i.e., $\frac{d\Q}{d\P} \in \partial C(x^*;\Pi)$.
Utilizing path independence (Theorem~\ref{thm:properties}\eqref{split}), it follows that $\vec{0} \in \argmax_{x \in L^\infty} \E^\Q[x - C(x;\Pi - x^* + C(x^*;\Pi)\vec{1})\vec{1}]$, i.e., $\frac{d\Q}{d\P} \in \partial C(\vec{0};\Pi - x^* + C(x^*;\Pi)\vec{1})$.

It remains to show that $\argmax_{x \in L^\infty} \E^\Q[x - C(x;\Pi)\vec{1}] \neq \emptyset$ for $\Q \in \M_*$. 
Let $\A := \{z \in L^\infty \; | \; u(z) \geq u(\Pi)\} \subseteq L^\infty_+$ be the superlevel set at $u(\Pi)$. By quasiconcavity and upper-semicontinuity, $\A$ is a weak* closed and convex set.
Fix $y \in L^1_+$, then
\begin{align*}
\sup_{x \in L^\infty} \E[y(x - C(x;\Pi)\vec{1})] &= \sup_{x \in L^\infty} \{\E[yx] \; | \; u(\Pi-x) = u(\Pi)\} = \sup_{x \in L^\infty} \{\E[yx] \; | \; u(\Pi-x) \geq u(\Pi)\}\\
    &= \sup_{x \in L^\infty} \{\E[yx] \; | \; \Pi - x \in \A\} = \E[y\Pi] - \inf_{x \in \A} \E[yx].
\end{align*}
By $\A \subseteq L^\infty_+$, it follows that $\inf_{x \in \A} \E[yx] \geq \inf_{x \in L^\infty_+} \E[yx] = 0$. As a direct consequence, we also have that $\sup_{x \in L^\infty} \E[y(x - C(x;\Pi)\vec{1})] < \infty$ for any $y \in L^1_+$.
For ease of notation, define $\B := \{x \in L^\infty \; | \; \Pi - x \in \A\}$ is weak* closed and convex.
Further, we can define the associated indicator function $\delta(x) := \begin{cases} 0 &\text{if } x \in \B \\ \infty &\text{if } x \not\in \B\end{cases}$ (proper, weak* lower semicontinuous, and convex) and support function $\delta^*(y) := \sup_{x \in \B} \E[yx]$ (proper, weak* lower semicontinuous, and convex) for $y \in L^1$. By the prior finiteness of the supremum result, $\dom \delta^* \supseteq L^1_+$ and, in particular, $\operatorname{int}\dom\delta^* \supseteq L^1_* := \{y \in L^1_+ \; | \; \essinf y > 0\}$.
Therefore $\partial\delta^*(\frac{d\Q}{d\P}) \neq \emptyset$ since $\frac{d\Q}{d\P} \in \operatorname{int}\dom\delta^*$.
By the Fenchel-Young inequality, any subdifferential $x^* \in \partial\delta^*(\frac{d\Q}{d\P})$ is a maximizer of $\sup_{x \in L^\infty} \E^\Q[x - C(x;\Pi)\vec{1}]$ and the result follows.
\end{proof}

\section{Decentralized liquidity pooling}\label{sec:pooling}

Thus far within this work we have formally introduced the AMMs for prediction markets. Such structures as previously studied (in, e.g.,~\cite{chen2012utility,othman2012profit,othman2013practical}) are considered with fixed available liquidity with a central operator running the AMM. Herein we wish to explore aspects of a decentralized AMM for prediction markets. Specifically, we consider an AMM to be \emph{decentralized} if the liquidity pool is comprised of investments by diverse individuals and entities; furthermore, these investors can both add or remove liquidity at any time prior to the realization of the random event $\omega$, even after some bets have already been placed.

First, we wish to note that investing as a liquidity provider is ``simple'' prior to the opening of the market to trades.  Specifically, if an investor provides $\ell > 0$ of liquidity to the market, then they receive a payout at the conclusion of the market equal to a fraction $\ell/L$ (where $L > 0$ is the total initial market liquidity) of $\Pi(\omega)$. That is, liquidity provision is a bet with payoff dependent on the state of the market.

Second, if an investor wishes to either add or remove liquidity after the market has opened to trades then we take the idea that these are special kinds of bets that adjust dynamically with the state of the market.  We will focus our discussion on the liquidity provision case as selling a liquidity position acts similarly.
Key to the construction of this special trade is that providing liquidity to the market should reduce the cost of trading for any counterparty.  As such, we follow the idea from, e.g.,~\cite{bichuch2022axioms} that pooling must be taken so that the pricing oracles are unaffected. We will, however, consider this trade in a generic manner first and then propose the specific structure for the constancy of the pricing oracles.

Briefly, let $\ell \in L^\infty_+ \backslash \{\vec{0}\}$ be a position that the liquidity provider is willing to sell in exchange for a fraction $\alpha > 0$ of the liquidity remaining $\Pi$. In particular, the investor wishes to guarantee that this fraction of liquidity remaining is accurate \emph{after} she places her own bet, i.e., $x = \alpha (\Pi - x + \ell)$ for some $\alpha > 0$; recall that $\Pi-x+\ell$ is the liquidity that remains in the AMM after the liquidity provision $\ell$ is added, but after the bet $x$ is placed, where the prior liquidity position is $\Pi$.
Solving for $x = \frac{\alpha}{1+\alpha}(\Pi + \ell)$, we solve the inverse problem to determine the fraction of the liquidity that is being purchased, i.e., $0 = C\left(\frac{\alpha}{1+\alpha}(\Pi+\ell)-\ell;\Pi\right)$.
In other words, prior to the market opening, the payoff for investors who provide funding to the AMM was simply the proportion of reserves that they provided. In contrast, after the market opens and bets have been placed, this is no longer as straightforward.  Liquidity providers are still entitled to some portion of the terminal wealth that remains at the AMM (i.e., a proportion of $\Pi$), but that portion may no longer be a simple proportion $\ell/\Pi$ as this may no longer even be deterministic. Instead, the aforementioned calculation determines the deterministic proportion $\alpha$ that will be assigned to the liquidity provider.
As will be discussed in greater detail below, the payout from this liquidity position will dynamically adjust along with the liquidity of the AMM as new bets are made; the payoff of this bet is the $\alpha$ fraction of the final liquidity $\bar\Pi \in \domu$ after all bets are placed. It is this dynamic adjustment of the payout which makes this bet a liquidity provision, i.e., the dynamic adjustment takes the opposite position to any new incoming bet thereby decreasing price impacts and increasing market liquidity. 
\begin{proposition}\label{prop:pooling-cost}
Fix the pool size $\Pi \in \domu$ and new liquidity provision $\ell \in L^\infty_+ \backslash \{\vec{0}\}$. There exists a unique fraction of the pool size purchased $\alpha^*(\Pi,\ell) > 0$ such that 
\begin{equation}\label{eq:pooling}
0 = \bar C(\alpha^*(\Pi,\ell)) := C\left(\frac{\alpha^*(\Pi,\ell)}{1+\alpha^*(\Pi,\ell)} \Pi - \frac{1}{1+\alpha^*(\Pi,\ell)}\ell;\Pi\right). 
\end{equation}
Similarly, there exists a unique fraction of the pool size sold at $\alpha^*(\Pi,\ell) \in (-1,0)$ satisfying \eqref{eq:pooling} for $\ell \in L^\infty_- \backslash \{\vec{0}\}$ with $\Pi + \ell \in \domu$.
\end{proposition}
\begin{proof}
Let $\Pi \in \domu$ and assume that $\ell \in L^\infty_+ \backslash \{\vec{0}\}$; as the proof for $\ell \in L^\infty_- \backslash \{\vec{0}\}$ with $\Pi + \ell \in \domu$ is similar, we omit it here.
By (strong) continuity of $C$ (see Corollary~\ref{cor:lipschitz}), it immediately follows that $\bar C$ is continuous on $\alpha \in (-1,\infty)$.
We also have that $\bar C(0) = C(-\ell;\Pi)<0$ and $\lim\limits_{\alpha\to\infty} \bar C(\alpha) =  \lim\limits_{\alpha\to\infty} C(\Pi; \Pi)>0.$ The existence of $\alpha^*(\Pi,\ell) > 0$ now follows. Using the fact that $\frac{\d}{\d \alpha} \left(\frac{\alpha}{1+\alpha}\Pi-\frac{1}{1+\alpha}\ell\right) = \frac1{(1+\alpha)^2}(\Pi + \ell) \in \domu$, and therefore $\alpha \mapsto  \frac{\alpha}{1+\alpha}\Pi-\frac{1}{1+\alpha}\ell$ is strictly increasing, and the fact that $C(\cdot)$ is also strictly increasing (see Theorem~\ref{thm:properties}\eqref{convex-monotone}), we get that $\bar C(\cdot) $ is strictly increasing as well. Thus we get the desired uniqueness.
\end{proof}

Notably, and as briefly mentioned above, the bet $\frac{\alpha^*(\Pi,\ell)}{1+\alpha^*(\Pi,\ell)}(\Pi+\ell)$ is made in such a way that it will always rebalance to maintain a payout of $\alpha^*(\Pi,\ell)$ fraction of the remaining liquidity as new bets are made.  Let $\Pi^* := \frac{1}{1+\alpha^*(\Pi,\ell)}(\Pi+\ell)$ be the effective liquidity after pooling the liquidity. When the next trade occurs, the liquidity provider (through, e.g., a smart contract) needs to simultaneously update their position. That is, when a new bet $x \in L^\infty$ is made, the trader will be charged $C(x+y-\alpha\Pi^*;\Pi^*)$ where $y \in \domu$ is the new holdings for the liquidity provider, i.e., $y = \alpha^*(\Pi,\ell)(\Pi^*-x-(y-\alpha^*(\Pi,\ell)\Pi^*)+C(x+y-\alpha^*(\Pi,\ell)\Pi^*;\Pi^*)\vec{1})$ so that the liquidity provider maintains the $\alpha^*(\Pi,\ell)$ fraction of remaining liquidity after the bet occurs.
\begin{proposition}\label{prop:pooling-rebalance}
Fix the pool size $\Pi^* \in \domu$ after a liquidity provision (or withdrawal) of $\alpha^* > -1$ was made.
The cost of trading is provided by the mapping $x \in L^\infty \mapsto (1+\alpha^*) C(\frac{1}{1+\alpha^*}x;\Pi^*)$.
\end{proposition}
\begin{proof}
The cost of trading $x$ with the liquidity provision that rebalances to stay in line with the market is given by $\tilde C := C(x+y-\alpha^*\Pi^*;\Pi^*)$ where $y \in L^\infty$ is the rebalancing required for the liquidity provider.  That is, $y$ satisfies the equilibrium $y = \alpha^*(\Pi^* - x - (y-\alpha^*\Pi^*) + C(x+y-\alpha^*\Pi^*;\Pi^*)\vec{1})$. Fixing $\tilde C \in \R$, this implies $y = \alpha^*\Pi^* - \frac{\alpha^*}{1+\alpha^*}x + \frac{\alpha^*}{1+\alpha^*}\tilde C\vec{1}$.  Therefore the result follows from translativity of the cost fuction (see Theorem~\ref{thm:properties}\eqref{split})
\[\tilde C = C\left(x+\left(\alpha^*\Pi^*-\frac{\alpha^*}{1+\alpha^*}x+\frac{\alpha^*}{1+\alpha^*}\tilde C\vec{1}\right)-\alpha^*\Pi^*;\Pi^*\right) = C\left(\frac{1}{1+\alpha^*}x;\Pi^*\right) + \frac{\alpha^*}{1+\alpha^*}\tilde C.\]
\end{proof}

With this foundation for the interaction of the liquidity provision and trades, we now want to specify $\ell \in L^\infty_+ \backslash \{\vec{0}\}$ so that the pricing oracles are kept constant, i.e., $P^b(\cdot;\Pi) = P^b\Big(\cdot;\frac{1}{1+\alpha^*(\Pi,\ell)}(\Pi+\ell)\Big)$ and $P^a(\cdot;\Pi) = P^a\left(\cdot;\frac{1}{1+\alpha^*(\Pi,\ell)}(\Pi+\ell)\right)$. Providing liquidity in this way means that pooling does not influence the pricing measure quoted by the market.
\begin{corollary}\label{cor:pooling-price}
Fix $t > -1$ and let the liquidity provided (or withdrawn) be given by $\ell = t\Pi$.  The fraction of the pool provided (or withdrawn) is given by $\alpha^*(\Pi,t\Pi) = t$ and the pricing oracles $P^b,P^a$ are invariant to the liquidity provision ($t > 0$) or withdrawal ($t \in (-1,0)$). 
\end{corollary}
\begin{proof}
Recall $\alpha^*(\Pi,\ell)$ is the unique root of $\bar C(\alpha)$ as defined in~\eqref{eq:pooling}. It trivially follows from that construction that $\alpha^*(\Pi,t\Pi) = t$ is a root of $\bar C$.  Furthermore, and as a direct consequence, the pricing oracles after the liquidity provision or withdrawal are provided by $P^b\left(\cdot;\frac{1}{1+\alpha^*(\Pi,t\Pi)}(\Pi+t\Pi)\right) = P^b(\cdot;\Pi)$ and similarly for $P^a$.
\end{proof}
This swap is a fraction of the current pool for that same fraction of the final pool (along with any collected fees). If no trades occur, then the investor will get their initial investment back exactly.

\begin{remark}
Note that before the first trade is made, the liquidity pool is $\Pi := L\vec{1}$ for initial liquidity $L > 0$.  Therefore, adding fixed liquidity $\ell > 0$ adheres to the proportional rule of Corollary~\ref{cor:pooling-price}.
\end{remark}

It remains to prove that providing liquidity reduces the costs for counterparties to buy or sell bets. Fundamentally, this is what liquidity provision is intended to accomplish as it decreases the price impacts experienced by any bettor. This is proven in the following lemma by noting that $\Pi^* = \Pi$ under the proportional liquidity provision $\ell = t\Pi$.
\begin{lemma}\label{lemma:pool-cash}
Fix the bet $x \in L^\infty$.  The cost of purchasing this bet is nonincreasing in the liquidity provision $\alpha$, i.e., $\alpha \in (-1,\infty) \mapsto (1+\alpha)C(\frac{1}{1+\alpha}x)$ is nonincreasing. 
\end{lemma}
\begin{proof}
We use the convexity of $C$ from Theorem~\ref{thm:properties}\eqref{convex-monotone} to demonstrate that the function $\alpha \in (-1,\infty) \mapsto (1+\alpha)C(\frac{x}{1+\alpha})$ is nonincreasing. Let $-1<\alpha_1<\alpha_2$, and define $\lambda:=\frac{1+\alpha_1}{1+\alpha_2} \in (0,1)$. From the convexity of $C$ and $C(0) = 0$, the desired result follows from
\begin{align*}
(1+\alpha_2)C\left(\frac{x}{1+\alpha_2}\right) &= (1+\alpha_2)C\left(\lambda(\frac{x}{1+\alpha_1})\right) \leq \lambda (1+\alpha_2) C\left(\frac{x}{1+\alpha_1}\right) = (1+\alpha_1) C\left(\frac{x}{1+\alpha_1}\right).
\end{align*}
\end{proof}

\begin{remark}
Before concluding this section, we want to briefly discuss the provision of fixed liquidity, i.e., $\ell = \bar\ell \vec{1}$ for some $\bar\ell > 0$. 
Though tempting, as this is an injection of (a fixed quantity of) cash to the market, this provision will (generally) alter the pricing oracles. In modifying the pricing oracles, it can be that a bet $x \in L^\infty$ becomes more expensive after the deposit of $\bar\ell > 0$ than before.
This is counter to the notion of a liquidity provision; we further note that cryptocurrency AMM markets explicitly accept liquidity only so that the pricing oracles are kept invariant to the change in pool sizes~\cite{bichuch2022axioms}.
\end{remark}

\section{Fee structure}\label{sec:fees}

Throughout this work we have considered the market making problem without any explicit fees that are being collected by the market maker.\footnote{Previously we have only considered implicit fees through the bid-ask spread (as in \cite{ostrovsky2009information}).}  In this section we aim to define a manner in which the AMM can collect fees while retaining the meaningful properties of the no-fee setting as encoded in Section~\ref{sec:amm-properties} above.

Conceptually, due to the translativity of the cost function $C$, the AMM is indifferent to any almost sure bet. As such, we propose that fees are only collected on the (positive) \emph{random} portion $x - \essinf x\vec{1}$ of the bet $x \in L^\infty$.  That is, for the fixed fee level $\gamma \geq 0$, the cost of purchasing $x \in L^\infty$ in the pool $\Pi \in \domu$ is given by
\begin{equation}\label{eq:fees}
C_\gamma(x;\Pi) := (1+\gamma)C(x-\essinf x\vec{1};\Pi) + \essinf x = (1+\gamma)C(x;\Pi) - \gamma \essinf x.
\end{equation}
In this way, the AMM collects $\gamma C(x-\essinf x\vec{1};\Pi)$ in cash to be paid out to the liquidity providers. This surplus is collected by the AMM poolers in compensation for the liquidity they provide. The no-fee setting corresponds exactly to $\gamma = 0$; as will be discussed below, for a well-functioning market, we will want to cap the fees at $\gamma \leq 1$ (see Remark~\ref{rem:fees-bound}).
We recall from the original setup in \eqref{eq:CI} that the cost function is only applied to the random part of a bet with the essential infimum being accounted for in the liquidity; as such this fee structure matches the logic applied to AMMs generally.
Note that the fees $\gamma$ introduced here are not endogenous to the model, but rather are arbitrarily set. Within the numerical case study of Section~\ref{sec:superbowl-dynamic}, we consider how varying $\gamma$ may alter market dynamics and, therefore also, the total collected fees. 
\begin{remark}
The fee structure~\eqref{eq:fees} corresponds with that of cryptocurrency AMMs as encoded in, e.g.,~\cite{angeris2020improved,lipton2021automated} in which the trader pays (a fixed fraction of) the assets they sell; herein the bettor is ``selling'' cash in exchange for the random payoff. We note that the notation utilized here differs as, traditionally for cryptocurrency AMMs, the inverse $C^{-1}(y;\Pi)$ is considered as the primitive instead.

We further wish to note that collecting fees as a fraction of the purchased bet results in random fees rather than a deterministic amount. Due to this, we opt solely to investigate the assessment of fees on the cost $C$ rather than the size of the bet $x \in L^\infty$.
\end{remark}

\begin{remark}
By defining the fees on the (positive) random portion $x - \essinf x\vec{1}$ of the bet $x \in L^\infty$, we eliminate a possible violation of the law of one price. Specifically, the cost of buying $x \in L^\infty_+$ is identical to the cost of buying the fixed payout of $\esssup x\vec{1}$ and selling $\esssup x\vec{1} - x$. 
Neglecting counterparty risks, the payoff of these bets are functionally identical and, in assessing fees only on the random portion of a bet, so are their costs.
\end{remark}

\begin{proposition}
Let $u: \domu \to \R$ be an LBAMM with remaining liquidity $\Pi \in \domu$ and associated payment function $C: L^\infty \to \R$.  For any fixed bet $x \in L^\infty\backslash\{c\vec{1} \; | \; c \in \R\}$, the mapping $\gamma \in \R_+ \mapsto C_\gamma(x;\Pi)$ is strictly increasing in fee level. Furthermore, $C_\gamma(c\vec{1};\Pi) = c$ for any $c \in \R$ and $\gamma \geq 0$.
\end{proposition}
\begin{proof}
Assume $x \in L^\infty$. By construction, $\frac{\d}{\d\gamma} C_\gamma(x;\Pi) = C(x-\essinf x \vec{1};\Pi)$. By Theorem~\ref{thm:properties}\eqref{bounds}, this is strictly positive so long as $x \not\in \{c\vec{1} \; | \; c \in \R\}$. Furthermore, $C_\gamma(c\vec{1};\Pi) = c$ for any $c \in \R$ by construction for any $\gamma\ge0.$
\end{proof}

Herein, we assume that the collected fees $\gamma C(x-\essinf x \vec{1};\Pi) \geq 0$ are immediately distributed to the liquidity providers. As such, when a bet is placed, the pool still updates from $\Pi \in \domu$ to $\Pi - x + C(x;\Pi)\vec{1} \in \domu$ as in the no-fee setting. By collecting and disbursing fees in this manner, the impact on our discussion of liquidity pooling in Section~\ref{sec:pooling} remains unaffected.

\begin{corollary}\label{cor:fees-properties}
Let $u: \domu \to \R$ be an LBAMM with remaining liquidity $\Pi \in \domu$ and associated payment function $C: L^\infty \to \R$. Fix the fee level $\gamma \geq 0$ and define $C_\gamma: L^\infty \to \R$ as in~\eqref{eq:fees}. It then satisfies the following properties:
\begin{enumerate}
\item \textbf{No arbitrage:} $C_\gamma(x) \in [\essinf x , (1+\gamma)\esssup x - \gamma \essinf x]$; 
\item \textbf{Increasing utility:} $u(\Pi - x + C_\gamma(x)\vec{1}) > u(\Pi)$ for $\gamma > 0$ and $x \in L^\infty \backslash \{c\vec{1} \; | \; c \in \R\}$ with strict inequality becoming an equality if $\gamma = 0$ or $x = c\vec{1}$ for some $c \in \R$;
\item \textbf{Convex and monotone:} $(z,c) \in \{(z,c) \in L^\infty_+ \times \R \; | \; \essinf z = 0\} \mapsto C_\gamma(z+c\vec{1})$ is convex and strictly increasing;\footnote{We interpret $(z,c) \in L^\infty_+ \times \R$ with $\essinf z = 0$ as the random portion $z$ and the constant portion $c$ of a bet. Notably $L^\infty = \{z+c\vec{1} \; | \; z \in L^\infty, c \in \R, \essinf z = 0\}$ so that this domain does not restrict the space of bets under consideration.  } 
\item \textbf{(Lower semi-)Continuous:} $x \mapsto C_\gamma(x)$ is lower semi-continuous in the weak* topology and Lipschitz continuous with Lipschitz constant 1; and
\item \textbf{Path independent:} 
    $C_\gamma(x+y;\Pi) = C_\gamma(x;\Pi) + C_\gamma(y;\Pi - x + C(x;\Pi)\vec{1})$ for $x,y \in L^\infty$ such that $\essinf[x+y] = \essinf x + \essinf y$. As a direct consequence, this path independence applies to splitting a trade $x = \lambda x + (1-\lambda)x$ for any $\lambda \in [0,1]$ and translativity $C_\gamma(x+c\vec{1};\Pi) = C_\gamma(x;\Pi) + c$ for any $x \in L^\infty$ and $c \in \R$.
\end{enumerate}
\end{corollary}
\begin{proof}
First, recall the definition of $C_\gamma$ from~\eqref{eq:fees}.
\begin{enumerate}
\item The bounds follow directly from Theorem~\ref{thm:properties}\eqref{bounds} applied to the bet $x-\essinf x\vec{1}$.
\item By monotonicity of $u$ and Theorem~\ref{thm:properties}\eqref{equal}, it follows for any $x \in L^\infty \backslash \{c\vec{1} \; | \; c \in \R\}$ and $\gamma > 0$ that
    \[u(\Pi - x + C_\gamma(x)\vec{1}) > u(\Pi - x + C(x)\vec{1}) = u(\Pi).\]
\item Note that $C_\gamma(z+c\vec{1}) = (1+\gamma)C(z) + c$ for any $(z,c) \in L^\infty \times \R$. In particular, this holds for $z \in L^\infty_+$ such that $\essinf z = 0$.
    \begin{enumerate}
    \item Let $(z,c) \gneq (\tilde z,\tilde c)$.  Then, by Theorem~\ref{thm:properties}\eqref{convex-monotone}, $C(z) \geq C(\tilde z)$ and $c \geq \tilde c$ with at least one of the inequalities strict, i.e., $C_\gamma$ has the desired strict monotonicity property. 
    \item Convexity follows directly from the construction of $C_\gamma(z+c\vec{1})$ and convexity of $C$ (see Theorem~\ref{thm:properties}\eqref{convex-monotone}).
    \end{enumerate}
\item Weak* lower semi-continuity follows directly from Theorem~\ref{thm:properties}\eqref{convex-monotone} as the essential infimum is weak* upper semi-continuous. Lipschitz continuity follows following the same logic as in Corollary~\ref{cor:lipschitz} (using the subsequent result on path independence to guarantee translativity).
\item 
    Path independence follows directly from Theorem~\ref{thm:properties}\eqref{split} applied to the definition of $C_\gamma$ under the assumption that $\essinf[x+y] = \essinf x + \essinf y$. The two consequences follow since $\essinf x = \essinf \lambda x + \essinf (1-\lambda) x$ and $\essinf[x+c\vec{1}] = \essinf x + c$ for any $x \in L^\infty$ and $c \in \R$.
\end{enumerate}
\end{proof}

\begin{remark}
\begin{itemize}
\item Note that the monotonicity of $C_\gamma$ with fees $\gamma > 0$ is more delicate than in the no-fee setting. This is due to the way in which fees are collected on the random portion of the bet only. Therefore, the monotonicity of costs need to be assessed separately on the random portion $z \in L^\infty_+$ (with $\essinf z = 0$) and the constant shift $c \in \R$. In fact, as observed by the upper bound on $C_\gamma(x)$, it is possible that there exists some bet $x \in L^\infty$ so that $C_\gamma(x) > \esssup x$ which would violate the naive attempt at monotonicity of $x \mapsto C_\gamma(x)$. For instance, consider the LBAMM from Example~\ref{ex:uniswap} with $N = 2$; if $x(\omega_1) = 0,~x(\omega_2) = 1,~\Pi(\omega_1) > 0,~\Pi(\omega_2) \in (0,\gamma\Pi(\omega_1) + \frac{\gamma}{1+\gamma})$ then $C_\gamma(x;\Pi) > 1 = \esssup x$. In fact, setting $\Pi(\omega_2) = \gamma \Pi(\omega_1)$ in that same example guarantees that $C_\gamma(x;\Pi)$ is (strictly) decreasing as $x(\omega_1) \in [0,1)$ increases. 
\item Though convexity of $C_\gamma$ in Corollary~\ref{cor:fees-properties} was only stated on the domain $\{z+c\vec{1} \; | \; z \in L^\infty, c \in \R, \essinf z = 0\}$, this can readily be shown to be equivalent to convexity of $x \in L^\infty \mapsto C_\gamma(x)$.
\item Note that path independence is now only defined for trades with additive essential infima. 
The general case for path independence, which holds in the no-fee setting as encoded in Theorem~\ref{thm:properties}\eqref{split}, is not desirable when fees are assessed. This becomes clear when considering a round trip trade when $x \in L^\infty_+$ is bought and then subsequently sold; path independence would imply such a trade nets $\$0$ for the liquidity providers which would violate the collection of any fees on the purchase or liquidation of the position.
\end{itemize}
\end{remark}

We wish to conclude our discussion of fees by considering how these fees will be quoted to users through the bid and ask prices, i.e., the modifications to the pricing oracles.  Specifically, we define the bid and ask prices, respectively, as (see \eqref{eq:price})
\begin{align*}
P^a_\gamma(x;\Pi) &:= \lim_{t \searrow 0} \frac{1}{t} C_\gamma(tx;\Pi) = (1+\gamma)P^a(x;\Pi) - \gamma \essinf x, \\
P^b_\gamma(x;\Pi) &:= \lim_{t \nearrow 0} \frac{1}{t} C_\gamma(tx;\Pi) = (1+\gamma)P^b(x;\Pi) - \gamma \esssup x 
\end{align*}
for any $\Pi \in \domu$. That is, buying a marginal unit of a bet with payoff $x$ will have a per unit cost of $P_\gamma^a(x;\Pi) = P^a(x;\Pi)+\gamma (P^a(x;\Pi) - \essinf x)$ which is increasing in $\gamma$. Further, selling a marginal unit of that bet will recover $P_\gamma^b(x;\Pi) = P^b(x;\Pi) - \gamma (\esssup x - P^b(x;\Pi))$ which decreases as $\gamma$ increases.
As expected, $P^b_\gamma(x;\Pi) \leq P^b(x;\Pi) \leq P^a(x;\Pi) \leq P^a_\gamma(x;\Pi)$ for every $x \in L^\infty$ and pool size $\Pi \in \domu$. However, as encoded here, the fees may not apply equally on both sides of the market, i.e., often $P^a_\gamma(x;\Pi) - P^a(x;\Pi) \neq P^b(x;\Pi) - P^b_\gamma(x;\Pi)$. 

\begin{remark}\label{rem:fees-bound}
Noting that a bettor will never elect to buy a bet costing more than its (essential) supremum nor sell to recover less than its (essential) infimum, the pricing oracles have the effective bounds of $P^a(x) \leq \frac{\esssup x + \gamma \essinf x}{1+\gamma}$ and $P^b(x) \geq \frac{\gamma \esssup x + \essinf x}{1+\gamma}$ for any $x \in L^\infty$ and $\gamma \geq 0$. In this way the ask price (including fees when $\gamma > 0$) is bounded from above by the essential supremum and the bid price is bounded from below by the essential infimum (i.e., $P^a_\gamma(x) \leq \esssup x$ and $P^b_\gamma(x) \geq \essinf x$). Implicitly to this construction, and necessary to be assumed in practice, the fees must therefore be bounded $\gamma \in [0,1]$ so that the bid and ask prices are ordered properly to follow financial logic.
\end{remark}

\section{Case Studies}\label{sec:cases}

Within this section, we wish to explore two case studies to explore the versatility and applicability of the AMMs constructed within this work. First, we will replicate a two outcome sports book with data collected for Super Bowl LVII. 
With this empirical case study, we explore the potential profits and losses accrued by the liquidity providers. 
In the other case study, we explore the use of an AMM for financial derivatives by simulating a system with a continuous probability space; in doing so, we prove the viability of our system to adjust the quoted distribution to investor actions.

\subsection{Super Bowl LVII}\label{sec:superbowl}

Sports betting on the Super Bowl is big business, with \$16 billion wagered in 2023 on Super Bowl LVII alone.\footnote{\url{https://www.espn.com/chalk/story/\_/id/35607249/survey-record-504-million-adults-bet-16b-super-bowl}} One popular way to bet on a single game is with the \emph{money line}. As opposed to the probabilities quoted within this work, the money line quotes the profits that would be gained from a winning bet of \$100 (if the underdog) (i.e. if the money line is $m>0$, then betting $100$, will get $m+100$ in case the bid wins), or how much needs to be bet to win \$100 (if the favorite) (i.e. if the money line is $m<0$, betting $-m$ will obtain $-m+100$ in case the bid wins). Therefore, it is easy to reformulate the money line as probabilities. Specifically, if $m$ is the quoted money line for one team to win, then the \emph{ask} probability is $P^a = \frac{-m\ind_{\{m < 0\}} + 100\ind_{\{m \geq 0\}}}{|m| + 100}$.

Consider the possible outcomes for Super Bowl LVII played between the Kansas City Chiefs (KC) and the Philadelphia Eagles (PHI), i.e., $\Omega = \{\KC,\PHI\}$ where we denote the events based on the eventual winner. In the 2 weeks leading off to the start of the game (i.e., once the two final teams are determined), bets are accepted at numerous sports books on this two-event sample space. For the purposes of this data, we collected money line data from Bookmaker.\footnote{Made available from \url{https://pregame.com/game-center/193165/odds-archive}} The time series of quoted (implied) bid and ask probabilities for KC are displayed within Figure~\ref{fig:superbowl-price}. (Note that the bid and ask probabilities for PHI are such that $P^b(\ind_{\PHI}) = 1 - P^a(\ind_{\KC})$ and $P^a(\ind_{\PHI}) = 1 - P^b(\ind_{\KC})$ respectively.)

\subsubsection{Deterministic Backtesting}\label{sec:superbowl-static}
Using this bid-ask spread, we were able to determine the (implied) mid-price by normalizing the ask prices to sum to 1. This price oracle, $P(\ind_{\KC})$ as displayed in Figure~\ref{fig:superbowl-price}, drives our backtesting system. Within this construction we consider 2 utility functions: (i) the expected logarithm utility function as discussed in Example~\ref{ex:uniswap} and (ii) the Liquid StableSwap utility function ($u(x) = \E[\log(x)] + \lambda \log(\E[x])$ with $\lambda = 2$) introduced in \cite{bichuch2022axioms}. With these LBAMMs, we deduce the trades necessary to exactly replicate the mid-price at Bookmaker (without fees). In doing so, we make no assumptions on trader's behavior but rather assume that the quoted mid-price would, fundamentally, be reflected in any constructed market. 
Assuming $L > 0$ cash reserves was initially available to the AMM, the liquidity available at any time can be determined via
\[\Pi(\KC) = L\sqrt{P(\ind_{\PHI})/P(\ind_{\KC})} \quad \text{ and } \quad \Pi(\PHI) = L\sqrt{P(\ind_{\KC})/P(\ind_{\PHI})}\]
for the logarithmic utility and numerically for Liquid StableSwap.
Notably, for both of these LBAMMs, the liquidity available $\Pi$ scales linearly with the initial cash reserves $L$; for this reason we quote all profits and losses as a percentage of the initial cash reserves rather than as an absolute value.
The liquidity remaining, as a percentage of the initial cash reserves, is shown as a time series in Figure~\ref{fig:superbowl-liquidity}.
By computing the liquidity remaining at all times, it is also possible to determine the bets that are actualized. In this way, we can determine profits gained from fees on trades over time as well as the gains or losses based on the final outcome of the event. 

Consider, first, the logarithmic utility function.
In Figure~\ref{fig:superbowl-profit-KC} we see how these profits accumulate until the game outcome ($\KC$) is realized with $\gamma = 1\%$.\footnote{Bookmaker adjusts the money line less frequently than our LBAMM would, thus the profits quoted herein are a lower bound to those that would be collected in practice due to the accumulation of many small trades which introduce additional volatility.} The quoted bid-ask spread over time with this $1\%$ fee is shown in Figure~\ref{fig:superbowl-price}. Notably, in Figure~\ref{fig:superbowl-profit-KC} and Figure~\ref{fig:superbowl-profit-PHI}, we see that the final profits for the liquidity providers can depend significantly on the realized outcome of the game with approximately $2.3\%$ return just prior to the start of the game. Due to the victory of KC, this would jump to nearly $5.4\%$ return; in the counterfactual scenario of a PHI win, the liquidity providers would, instead, be subject to a nearly $0.7\%$ loss. If the fees were increased to only $1.3\%$ (an increase of just 30bps), then a liquidity provider would break even if PHI was realized and have over $6\%$ return from the KC victory.

In contrast, by selecting the Liquid StableSwap utility function, we are able to increase the fees collected from trading substantially.  In Figure~\ref{fig:superbowl-stable-KC} we see how these profits accumulate until the game outcome ($\KC$) is realized with $\gamma = 1\%$.\footnote{As we fixed the prices based on the data, this bid-ask spread is identical to that found with the logarithmic utility function.} From the fees alone, liquidity providers would experience almost a $6.5\%$ return prior to the start of the game.  Due to the victory of KC, this would jump to nearly $15.7\%$ return; in the counterfactual scenario of a PHI win, the liquidity providers would, instead, be subject to a $2.5\%$ loss. With these greater potential losses, the fees would need to increase to $1.38\%$ (an increase of just 38bps) to guarantee a liquidity provider breaks even when PHI was realized (and over $18.1\%$ return from the KC victory).

\begin{figure}[t]
\centering
\begin{subfigure}[t]{0.3\textwidth}
\centering
\includegraphics[width=\textwidth]{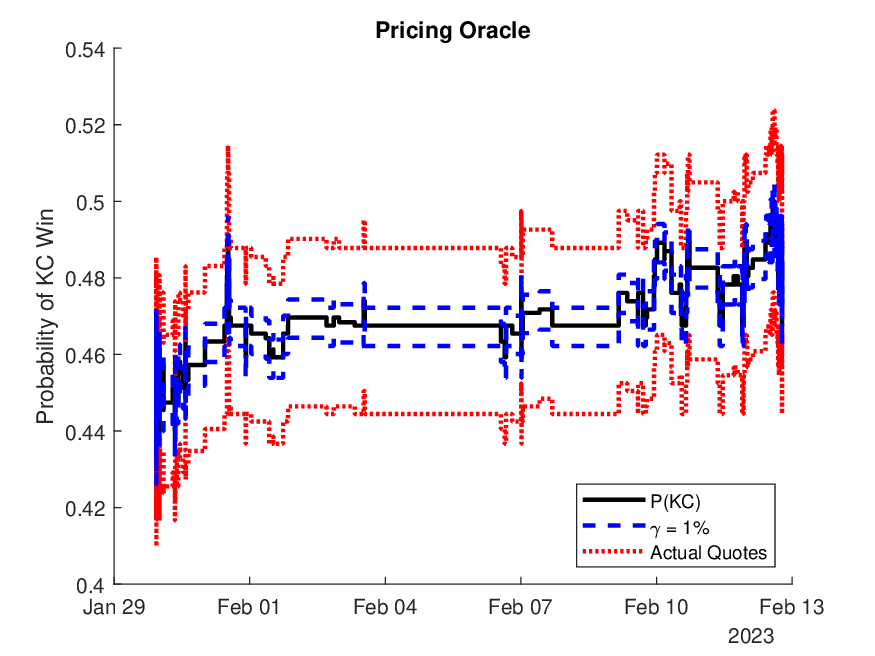}
\caption{Time series of the quoted probability and bid-ask prices for $\KC$.}
\label{fig:superbowl-price}
\end{subfigure}
~
\begin{subfigure}[t]{0.3\textwidth}
\centering
\includegraphics[width=\textwidth]{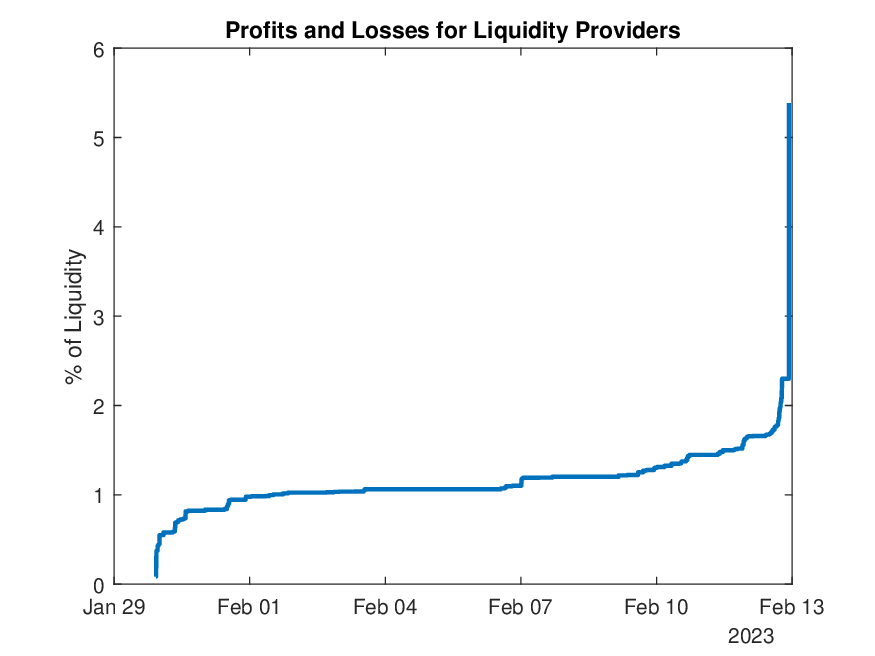}
\caption{Realized profits and losses for logarithmic utility exhibited from Super Bowl LVII.}
\label{fig:superbowl-profit-KC}
\end{subfigure}
~
\begin{subfigure}[t]{0.3\textwidth}
\centering
\includegraphics[width=\textwidth]{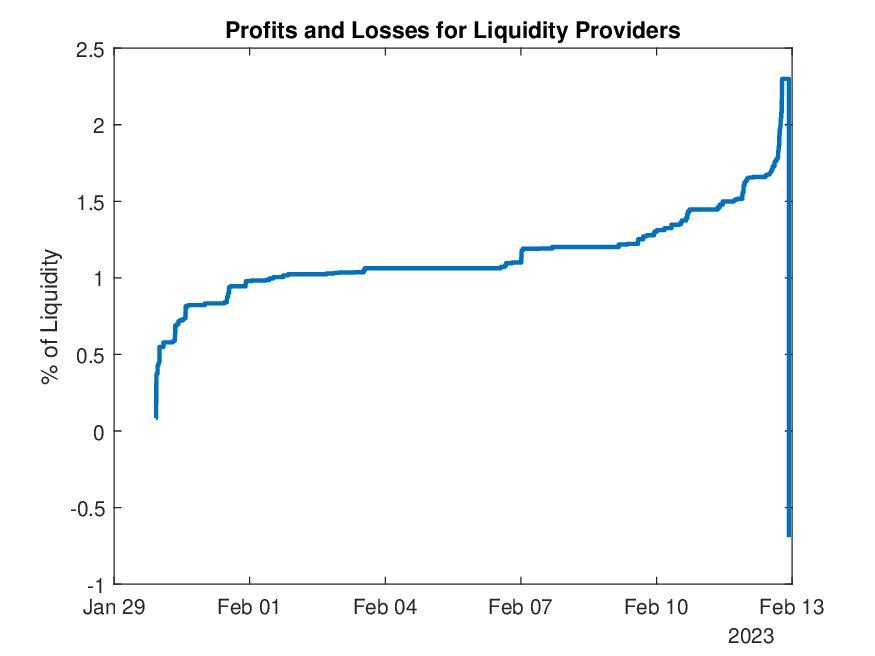}
\caption{Counterfactual profits and losses for logarithmic utility if PHI had won Super Bowl LVII.}
\label{fig:superbowl-profit-PHI}
\end{subfigure}
~
\begin{subfigure}[t]{0.3\textwidth}
\centering
\includegraphics[width=\textwidth]{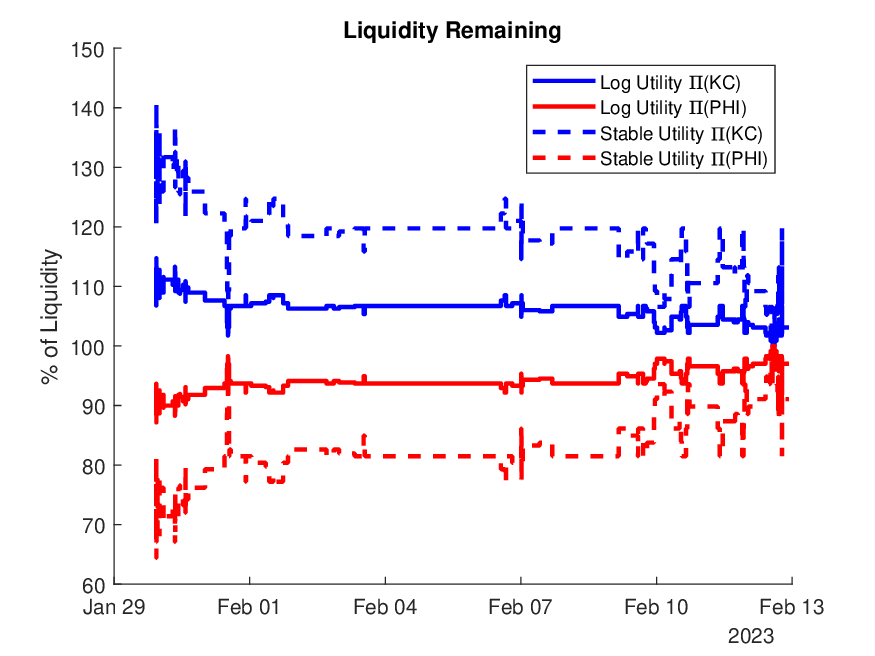}
\caption{Time series of the liquidity remaining with logarithmic and liquid stableswap utility.}
\label{fig:superbowl-liquidity}
\end{subfigure}
~
\begin{subfigure}[t]{0.3\textwidth}
\centering
\includegraphics[width=\textwidth]{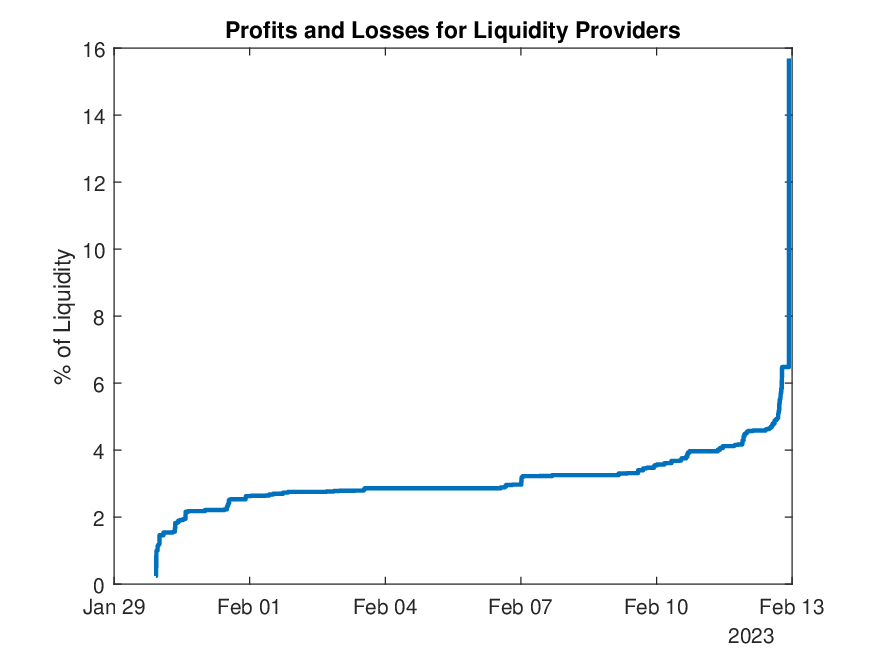}
\caption{Realized profits and losses for liquid stableswap utility exhibited from Super Bowl LVII.}
\label{fig:superbowl-stable-KC}
\end{subfigure}
~
\begin{subfigure}[t]{0.3\textwidth}
\centering
\includegraphics[width=\textwidth]{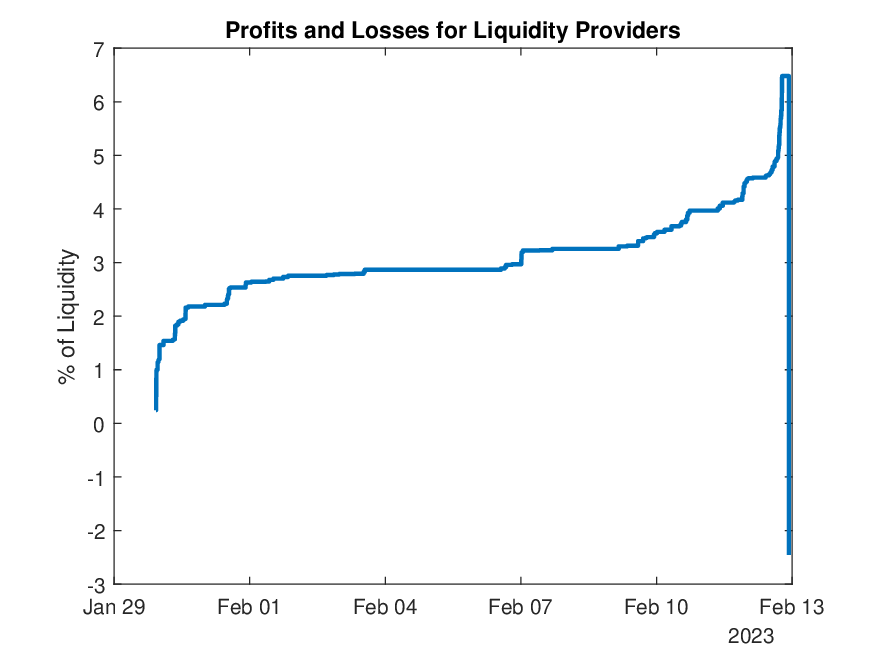}
\caption{Counterfactual profits and losses for liquid stableswap utility if PHI had won Super Bowl LVII.}
\label{fig:superbowl-stable-PHI}
\end{subfigure}
\caption{Section~\ref{sec:superbowl-static}: Visualization of LBAMMs when applied to Super Bowl LVII.}
\label{fig:superbowl}
\end{figure}

\subsubsection{Stochastic Backtesting}\label{sec:superbowl-dynamic}
We now aim to introduce a stochastic backtesting framework to the Bookmaker money line data we have collected. In contrast to the deterministic backtest of Section~\ref{sec:superbowl-static} in which the AMM's price oracle exactly follows the mid-price of the external sports book, here we assume that the underlying market price follows a Brownian motion, starting from the initial mid-price of the external market. We constrain the price process to remain within Bookmaker's bid-ask spread through a reflection principle. However, rather than assuming the AMM perfectly tracks this price process, we assume trades occur only to capture arbitrage opportunities between the true (stochastic) price process and the price oracle when accounting for the fees $\gamma$. In this way, as in Section~\ref{sec:superbowl-static}, we make no assumptions on bettor behavior but rather utilize the available betting line data to construct simple price processes in order to find (potential) arbitrage opportunities.

Due to the profits exhibited in the prior backtest, we focus exclusively on a dynamic version of the Liquid StableSwap AMM herein with $u_t(x) = \E^{\P_t}[\log(x)] + \lambda \log(\E^{\P_t}[x])$ with $\lambda = 2$ and $\P_t(KC)$ is determined by the (implied) mid-price of Bookmaker money line data at time $t$. The choice of dynamic probability measure is such that the stable region of the AMM tracks the mid-price of the external market and is intended to maximize the fees collected. As far as the authors are aware, such a dynamic AMM has not previously been proposed in either prediction or cryptocurrency markets.

Figure~\ref{fig:superbowl-dynamic} displays the expected profits, assuming the terminal mid-price of the external data is the true probability of a KC victory, under varying fee levels $\gamma \in \{0\%,0.5\%,1\%,...,5\%\}$ and (annual) volatilities of the Brownian motion  $\sigma \in \{5\%,25\%,50\%\}$. To complete these computations we consider Monte Carlo simulations with 500 price paths (with a time step of 1 minute); both the average and the 95\% confidence interval are plotted in Figure~\ref{fig:superbowl-dynamic}.
The common pattern in these profits across different volatilities $\sigma$ is due to the competing elements insofar as the number of trades decreases as the fees $\gamma$ increase as there are fewer arbitrage opportunities to exploit. Thus there is a tradeoff between the number of transactions and the fees collected per transaction. Notably, at the extremes, no fees are collected when $\gamma = 0$ and no trades occur when $\gamma = 5\%$, as this fee level marginally exceeds the bid-ask spread of Bookmaker.

Consider now the optimal fee levels for the liquidity provider. Considering Figure~\ref{fig:superbowl-dynamic}, though subtle, it can be seen that the optimal fee level $\gamma$ is decreasing in volatility. That is, as volatility increases, the greater number of transactions dominates the aforementioned tradeoff. 
Intuitively, higher volatility increases the probability that the AMM's price sufficiently deviates from the Brownian price path to create arbitrage opportunities. In opposition, and as already noted, higher fees reduce the trading activity by widening the no-arbitrage range. Therefore, when volatility is higher, increased potential betting activity around the Brownian price path dominates the reduction in fees collected from each individual transaction.
Furthermore, these optimal fee levels (roughly around $\gamma \approx 1\%$) -- which result in significant expected profits of between 1\% and 6\% return in a 2 week period -- come at a much lower bid-ask spread than that quoted by Bookmaker. 
Consequently, the introduction of the AMM creates a win-win situation: liquidity providers can optimize their expected profits and liquidity takers find a more efficient market in which to transact. 
By the same arguments, the optimal fees are monotonic with respect to Bookmaker's bid-ask spread but are always dominated by those external fees. In this way, the optimal AMM fee construction must always increases market efficiency. 

\begin{figure}
\centering
\includegraphics[width=0.6\textwidth]{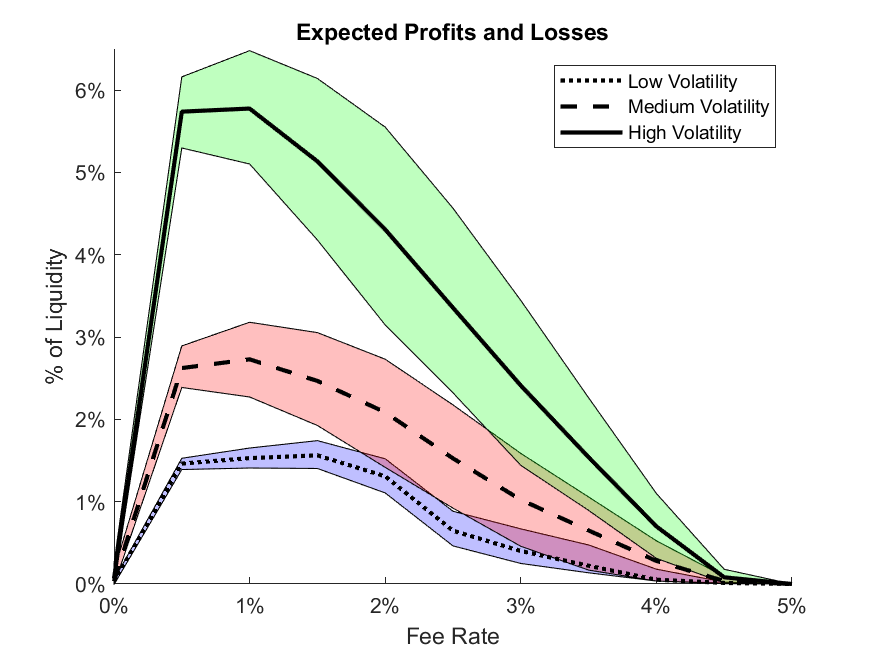}
\caption{Section~\ref{sec:superbowl-dynamic}: Expected profits (average and 95\% confidence interval) for dynamic liquidity stableswap utility on simulated betting data calibrated to Super Bowl LVII.}
\label{fig:superbowl-dynamic}
\end{figure}

\subsection{Financial Derivatives}\label{sec:derivative}
In this final case study, we explore the use of an LBAMM for pricing European options with a fixed expiry. Notably, for this construction, we want to consider a general probability space rather than the finite probability space considered in the prior case study. In doing so, we consider the LBAMM based on the scaled and shifted logarithmic utility function as is considered in Example~\ref{ex:essinf} with $\epsilon = 10^{-6}$ chosen arbitrarily. This AMM is constructed with initial cash reserves of $L = 100$ distributed so as to create an initial lognormal distribution for the price at the maturity time.

Herein we consider a derivatives payoff as a bet on a market outcome. The LBAMM then represents a type of utility indifference pricing, see, e.g., \cite{carmona2008indifference}. In our case, an LBAMM is indifferent between two scenarios -- one in which no bet is placed, and the other in which a bet is placed with the liquidity increased by the bet's price. 
However, in contrast to a typical derivatives market, we do not consider the market for the underlying asset(s), i.e., it is not possible to hedge using the underlying securities. Instead, the initial liquidity is used to compensate for the fact that there is no hedge possible for the market makers; this initial liquidity is used to guarantee that the LBAMM can pay off the sold derivatives at maturity.
That is, within this framework without the underlying market for hedging, options and other derivatives can be viewed as classical bets on market outcomes.

In order to avoid this derivatives market from converging to a Dirac measure, we assume that market trading ends a fixed amount of time prior to expiry. Due to the liquidity-bounded loss for the LBAMM, this market is able to trade any measurable payoff structure $x \in L^\infty$. In Figure~\ref{fig:derivative}, the price density is plotted under three circumstances: (i) the initial lognormal distribution; (ii) after 50 put options with strike at \$1 have been purchased (at an average cost of \$0.1103 per contract); and (iii) after 100 put options with strike at \$1 have been purchased (at an average cost of \$0.1166). It can clearly be seen that these derivative purchases appropriately shift the mass of probability leftward and increase its peakedness.  In addition, though subtle, there is a kink in the distribution at \$1 to match up with the strike price used. 
Finally, we would like to point out that in the lognormal setup used for this case study, we can compare these per-contract costs to the Black-Scholes price of \$0.0997. This price corresponds to the cost of a hedge under a risk-neutral measure. This should be compared to the price we obtained earlier, which represents the utility indifference price in a scenario where no hedging is allowed.
\begin{figure}
\centering
\includegraphics[width=0.6\textwidth]{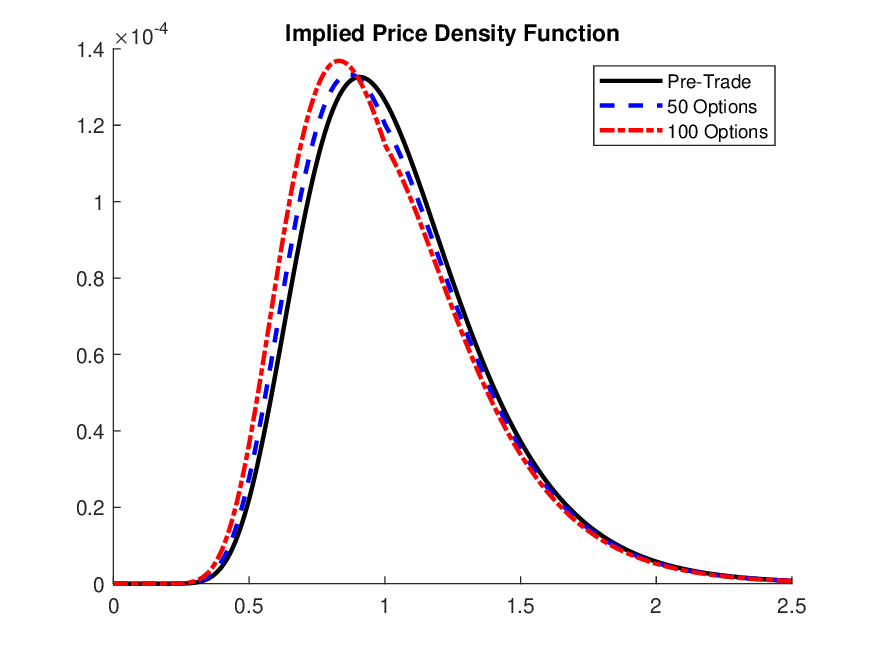}
\caption{Section~\ref{sec:derivative}: Impact of purchasing put options on the quoted pdf for market prices at maturity.}
\label{fig:derivative}
\end{figure}

While put options naturally have uniformly bounded payoffs, that is not true for all options. For instance, call options have unbounded payoff and, thus, do not fit into the $L^\infty$ framework proposed within this work.
However, much as in the risk measurement literature (see, e.g., \cite{delbaen2002coherent}), call options and other derivatives with unbounded payouts can be adapted to this framework by imposing a cap to the bet, e.g., the adapted call option would take value $x_T := \min\{(S-K)^+ \, , \, T\}$ for underlying price $S$, strike price $K$, and some cap $T > 0$ sufficiently large. 
Figure~\ref{fig:call-fixedL} displays the cost of 100 at-the-money call options with stirke \$1 assuming an initial lognormal distribution; notably, if $T$ is set too high, e.g., so that $T > \esssup\Pi$, then the cost of these options will be proportional to $T$ to guarantee the positivity of the resulting position $\Pi - x_T + C(x_T;\Pi)\vec{1} \in L^\infty_*$. 
However, if we scale the initial liquidity with the cap $T$ then, due to the low probability of high payouts, the cost of these 100 call options begins to fall for $T$ (and therefore liquidity) large enough; this is displayed in Figure~\ref{fig:call-propL}.
\begin{figure}
\centering
\begin{subfigure}[t]{0.45\textwidth}
\centering
\includegraphics[width=\textwidth]{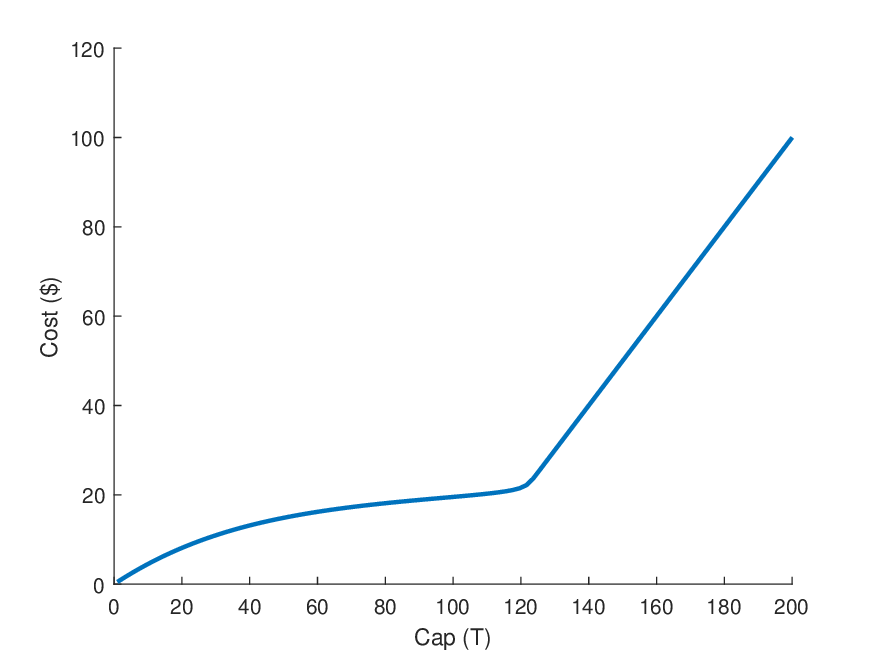}
\caption{Fixed initial liquidity $L = \$100$.}
\label{fig:call-fixedL}
\end{subfigure}
~
\begin{subfigure}[t]{0.45\textwidth}
\includegraphics[width=\textwidth]{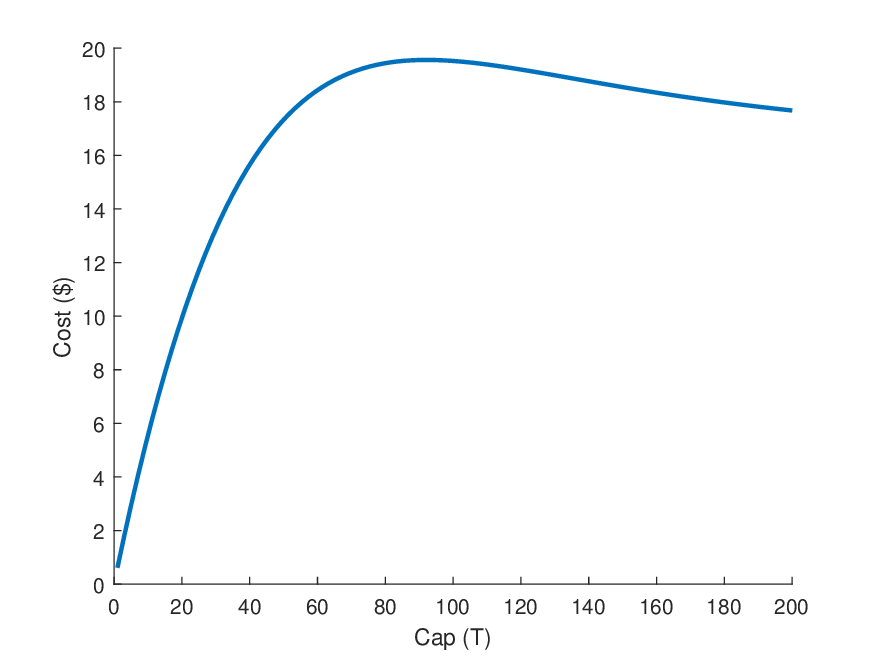}
\caption{Varying initial liquidity $L = \$T$.}
\label{fig:call-propL}
\end{subfigure}
\caption{Section~\ref{sec:derivative}: Impact of placing a cap on the cost of 100 at-the-money call options with strike \$1.}
\label{fig:call}
\end{figure}

\begin{remark}
For both of our numerical experiments with puts and call options, we have only considered the case with a fixed strike price. We note that LBAMMs are designed to permit user to select any desired strike price. In such a way, even when considering digital puts and calls (to guarantee boundedness and eliminate the need for a cap, as was done for our European call option example above), we utilize a general probability space to allow users to freely set the strike price.
\end{remark}

\section{Conclusion}\label{sec:conclude}
To summarize, in this manuscript, we have presented a general utility framework for a prediction market maker over general probability spaces. The novelty of this framework is that it considers the liquidity separately and in a decentralized way, which in turn allows for additional liquidity to be provided or withdrawn after the market opened and bets have already been received. We have also investigated the resulting properties of the pricing oracles. Additionally, we have proposed a novel way to charge fees on the random part of the bet that does not create an arbitrage.

Future research in decentralized prediction markets should more rigorously address the problem of optimizing fees so as to maximize (risk-adjusted) profits gained by the liquidity providers. Doing so will require dynamic models of betting so as to accurately study risks and returns. Additionally, in the proposed setup, we ignore the role of information on prediction markets; if an event becomes a certainty prior to the maturity of the bet, the liquidity providers would be arbitraged until no cash reserves remain. To avoid such a fate, we recommend a dynamic fee schedule which can counteract the informational gains of bettors. Such a system, especially to optimize this fee schedule, would be of great import for practical implementations.
\section*{Conflicts of Interest}
We, the authors, confirm that we have no relevant or material financial interests that relate to the research in this paper.
\section*{Data Availability}
The data that support the findings of this study are available from the any of the authors upon reasonable request.
\section*{Funding Information}
Maxim Bichuch and Zachary Feinstein are partially supported by Stellar Development Foundation Academic Research Grants program.\\ 
Zachary Feinstein acknowledges the support from NSF IUCRC CRAFT center research grant (2113906) for this research. The opinions expressed in this publication do not necessarily represent the views of NSF IUCRC CRAFT.

\bibliographystyle{acm}
\small{\bibliography{bibtex2}}

\end{document}